\documentclass{emsprocart}

\contact[dmitry.chelkak@ens.fr]{Dmitry Chelkak, {holder of the ENS--MHI chair funded by MHI}\\ 
D\'epartement de math\'ematiques et applications de l'ENS,\\
Ecole Normale Sup\'erieure PSL Research University, CNRS UMR 8553, Paris 5\`eme.\\[2pt]
{On leave from St.~Petersburg Department of Steklov Mathematical Institute RAS.}}

\numberwithin{equation}{section}

\newtheorem{theorem}{Theorem}[section]

\newtheorem{lemma}[theorem]{Lemma}
\newtheorem{proposition}[theorem]{Proposition}

\theoremstyle{definition}
\newtheorem{definition}[theorem]{Definition}
\newtheorem{remark}[theorem]{Remark}

\newcommand\wind{\mathrm{wind}}
\newcommand\Pf{\mathrm{Pf}}
\newcommand\sign{\mathrm{sign}}
\newcommand\Corr[2]{\langle#2 \rangle_{#1}}
\newcommand\opsi{\psi^{\star}}

\newcommand\Cdiscr{\mathbb{C}^{\textstyle\diamond}}

\title[2D Ising model: correlations via boundary value problems]{2D Ising model: correlation functions at criticality via Riemann-type boundary value problems}

\author[Dmitry Chelkak]{Dmitry Chelkak\thanks{The author is grateful to the University of Geneva, NCCR SwissMAP of the Swiss NSF, and ERC AG COMPASP for hospitality and support during the academic year 2015/16.}}

\begin{document}

\begin{abstract} In this note we overview recent convergence results for correlations in the critical planar nearest-neighbor Ising model. We start with a short discussion of the combinatorics of the model and a definition of fermionic and spinor observables. After that, we illustrate our approach to spin correlations by a derivation of two classical explicit formulae in the infinite-volume limit. Then we describe the convergence results (as the mesh size tends to zero, in arbitrary planar domains) for fermionic correlators~\cite{chelkak-smirnov-12}, energy-density~\cite{hongler-smirnov-13} and spin expectations~\cite{chelkak-hongler-izyurov-15}. Finally, we discuss scaling limits of mixed correlators involving spins, disorders and fermions, and the classical fusion rules for them. \end{abstract}

\begin{classification} Primary 82B20; Secondary 30G25, 81T40
\end{classification}

\begin{keywords} Ising model, conformal invariance, spin correlations, discrete holomorphicity
\end{keywords}

\maketitle

\section{Introduction}

The main goal of this note is to give a survey of convergence results for correlation functions in the critical planar Ising model obtained during the last several years. The Ising model, which is the simplest lattice model of a ferromagnet, was proposed by Lenz in 1920 and is now considered to be an archetypical example of a statistical mechanics system that admits an order-disorder phase transition in dimensions two and above, and for which the appearance of the \emph{conformal symmetry} at criticality in dimension two can be rigorously understood in great detail. Certainly, everybody knows that ``2D Ising model is a free fermion'' though this statement may look a bit vague for the probabilistic community. More precisely, the partition function of the \emph{nearest-neighbor} Ising model on a planar graph~$G$ can be written~\cite{hurst-green-60} as the Pfaffian of some matrix (e.g., indexed by oriented edges of~$G$). This fact allows one to introduce so-called \emph{fermionic observables} as the Pfaffians of (small-size) minors of the inverse matrix and give a concrete meaning to the statement mentioned above: if one interprets these observables as formal correlators, the fermionic Wick rule for the multi-point ones is built-in. Such observables satisfy simple linear equations which (at criticality) can be interpreted as a \emph{discrete holomorphicity} property and can be equivalently defined in a purely combinatorial manner~\cite{smirnov-10icm}. Moreover, as was proposed by Smirnov in his seminal papers~\cite{smirnov-06icm, smirnov-10}, they can be thought of as solutions to discrete versions of some special Riemann-type boundary value problems in order to prove their convergence to conformal covariant limits.

\smallskip

Nevertheless, it is worth noting that the fermionic observables per se do not allow one to analyze the \emph{spin correlations}, which are presumably the most interesting quantities appearing in the Ising model. An appropriate tool to study them is \emph{spinor observables}~\cite{chelkak-izyurov-13,chelkak-hongler-izyurov-15}, which can be thought of as generalizations of the fermionic ones for the Ising model considered on an appropriate double-cover of~$G$ and constrained with the spin-flip symmetry between the sheets. A more systematic way to introduce them is provided by the famous spin-disorder formalism of Kadanoff and Ceva~\cite{kadanoff-ceva-71}. In this language, the fermionic variables are obtained by fusing (a part of) spins and disorders, and the relevant Pfaffian identities can be deduced from the combinatorial representations of their correlators.

\smallskip

We review the combinatorics of the 2D Ising model in Section~\ref{section:combinatorics}, following~\cite{chelkak-cimasoni-kassel-16}. Note that one can define the fermionic and spinor observables in the Ising model considered on an \emph{arbitrary} planar graph, as well as use them for the study of the model away of criticality. We illustrate our approach to the analysis of spin correlations in Section~\ref{section:explicit-formulae}. Namely, we give a self-contained derivation of two classical results about the (critical and subcritical) diagonal spin-spin correlations in the full-plane using a direct link with the theory of orthogonal polynomials provided by spinor observables; see~\cite{chelkak-hongler-16} for similar computations in the half-plane. Section~\ref{section:conformal-invariance} follows \cite{chelkak-smirnov-12, hongler-smirnov-13,chelkak-hongler-izyurov-15,chelkak-hongler-izyurov-16} and is devoted to the convergence and conformal covariance of the correlation functions at criticality. For simplicity, we consider the Ising model on \emph{square grid} approximations of a given planar domain~$\Omega$; in fact, a good portion of the results can be directly generalized to isoradial graphs. We assume that~$\Omega$ is simply connected and consider~``$+$'' boundary conditions only; see~\cite{chelkak-hongler-izyurov-16} for a general setup.
The presentation is organized so as to highlight the correspondence between discrete objects and the standard Conformal Field Theory language used to describe the continuum limit of the critical Ising model. In particular, the normalizing factors in discrete are adjusted so as to fit the ones in continuum.

\smallskip

It should be said that there are plenty of important topics on the 2D Ising model that we do not touch in this note. There are more involved methods to study spin correlations in the infinite-volume limit, notably a link with Painlev\'e equations developed in~\cite{wu-mccoy-tracy-barouch-76}, quadratic identities found in~\cite{mccoy-perk-wu-80,perk-80} and the exact bosonization approach suggested in~\cite{dubedat-bosonization-11}; see also the monographs~\cite{mccoy-wu-book-14} and~\cite{palmer-book-07}. At criticality, one might be interested in convergence results for lattice counterparts of other CFT fields (e.g., the stress-energy tensor~\cite{chelkak-glazman-smirnov-16}) and in a definition of the Virasoro algebra action on these lattice fields~\cite{hongler-kytola-viklund-16}. Also, we do not touch the conformal invariance of curves~\cite{5-authors-14,izyurov-15} and loop ensembles~\cite{kemppainen-smirnov-15,benoist-duminil-copin-hongler-14,benoist-hongler-16} arising in the critical model. Finally, an important progress has been achieved recently~\cite{giuliani-grennblatt-mastropietro-12} in the analysis of the finite-range 2D Ising model via rigorous renormalization techniques.

\medskip

\noindent {\bf Acknowledgements.} First, I wish to thank my co-authors Cl\'ement Hongler and Konstantin Izyurov, to whom many of the ideas discussed in this note belong. It was also a great pleasure to collaborate with David Cimasoni, Alexander Glazman and Adrien Kassel on \cite{chelkak-cimasoni-kassel-16} and \cite{chelkak-glazman-smirnov-16}. In addition, I would like to thank Hugo Duminil-Copin, Kalle Kyt\"ol\"a and Wendelin Werner for many useful discussions and support. Last but not least, I am greatly indebted to Stanislav Smirnov, who introduced me into this field ten years ago and with whom I had the privilege to work on \cite{chelkak-smirnov-12}.

\newpage

\section{Combinatorics of the nearest-neighbor Ising model in 2D}
\label{section:combinatorics}

\subsection{Definition and contour representations of the planar Ising model}
\label{subsect:contours}
 Let~$G$ be a finite connected \emph{planar} graph embedded into the plane so that all its edges are straight segments.
The (ferromagnetic) \emph{nearest-neighbor} Ising model on the graph \emph{dual} to~$G$ is a random assignment of spins~$\sigma_u\in\{\pm 1\}$ to the~\emph{faces} of~$G$ with the probabilities of spin configurations~$\sigma\!=\!(\sigma_u)$ proportional to
\[
\textstyle \mathbb{P}_G[\,\sigma\,]\propto\exp\,[\,\beta \sum_{u\sim w} J_e\sigma_u\sigma_w\,]\,,\quad e=(uw)^*,
\]
where the positive parameter~$\beta$ is called the \emph{inverse temperature}, the sum is taken over all pairs of adjacent faces~$u,w$ (equivalently, edges~$e$) of~$G$, and~$J=(J_e)$ is a given collection of positive~interaction constants indexed by the edges of~$G$.

\smallskip

The~\emph{domain walls representation} (aka low temperature expansion) of the model is a~$2$-to-$1$ correspondence between spin configurations and even subgraphs~$P$ of~$G$: given a spin configuration,~$P$ consists of all edges such that the two adjacent spins differ from each other. We will often consider a decomposition of~$P$ into a \emph{collection of non-intersecting and non-self-intersecting loops}, note that it is not unique in general. Below we will always assume that the spin of the outermost face of~$G$ is fixed to be~$+1$, which is often described as \emph{``$+$\!''~boundary conditions}. Then the above correspondence becomes a bijection and one can write
\begin{equation}
\label{eq:spin-spin-low-temperature}
\textstyle \mathbb{E}_G[\sigma_{u_1}...\sigma_{u_m}]\;=\; \mathcal{Z}_G^{-1}\sum_{P\in\mathcal{E}_G}x(P)(-1)^{\mathrm{loops}_{[u_1,..,u_m]}(P)}\,,
\end{equation}
where~$\mathcal{E}_G$ denotes the set of all even subgraphs of~$G$,
\[
\textstyle \mathcal{Z}_G=\sum_{P\in\mathcal{E}_G} x(P)\,,\qquad x(P):=\exp[-2\beta\sum_{e\in P} J_e]\,,
\]
and~$\mathrm{loops}_{[u_1,..,u_m]}(P)$ is the number of loops in~$P$ surrounding an odd number of faces~$u_1,...,u_m$. If the graph~$G$ is not trivalent, this number is not uniquely defined (as there can be several ways to decompose~$P\in\mathcal{E}_G$ into a collection of loops) but it is always well defined modulo~$2$. The quantity~$\mathcal{Z}_G$ is called the~\emph{partition function} of the model. It is convenient to introduce the following parametrization:
\[
\textstyle x(P)=\prod_{e\in P}x_e\,,\qquad x_e=\tan\frac{1}{2}\theta_e:=\exp[-2\beta J_e]\,,
\]
where~$x_e\in [0,1]$ and~$\theta_e:=2\arctan x_e\in [0,\frac{1}{2}\pi]$ have the same monotonicity as~$\beta^{-1}$.

\smallskip

There exists another classical way of representing spin correlations (first observed by van der Waerden~\cite{van-der-warden-41} and known as the \emph{high temperature expansion}):
for the Ising model with spins assigned to~\emph{vertices} of~$G$ and interaction constants~$J_e^*$\,, cancellations caused by the fact that all products of spins are~$\pm 1$ imply the equality
\begin{equation}
\label{eq:spin-spin-high-temperature}
\mathbb{E}_G^*[\sigma_{v_1}...\sigma_{v_{2n}}] \; = \; (\mathcal{Z}_G^*)^{-1}{\textstyle\sum_{P\in\mathcal{E}_G(v_1,...,v_{2n})}x^*(P)}\,,
\end{equation}
where~$\mathcal{E}_G(v_1,...,v_{2n})$ denotes the set of subgraphs of~$G$ such that each of~$v_1,...,v_{2n}$ has an odd degree in~$P$ while the degrees (in~$P$) of all other vertices are even,
\[
\textstyle \mathcal{Z}_G^*=\sum_{P\in\mathcal{E}_G}x^*(P)\quad \text{and}\quad x^*(P):=\prod_{e\in P}\tanh [\beta^* J_e^*]\,.
\]

\begin{remark}
\label{rem:kramers-wannier}
It is well known that the homogeneous (all~$J_e\!=\!1$) Ising model on the square grid exhibits a second order \emph{phase transition} at~$\beta_\mathrm{crit}=\frac{1}{2}\log(1+2^{\frac{1}{2}})$: in the infinite-volume limit, there exists a unique Gibbs measure above and at the critical temperature~$\beta_{\mathrm{crit}}^{-1}$, while the subcritical model has two extremal ones describing~``$+$'' and~``$-$'' phases, respectively (e.g., see~\cite{aizenmann-80}). The explicit value of~$\beta_{\mathrm{crit}}$ can be found from the Kramers--Wannier~\cite{kramers-wannier-41} self-duality condition: if we assume that $x=x^*:=\tanh \beta^*$ and use the same parametrization~$\tan\frac{1}{2}\theta^*=\exp[-2\beta^*]$ for the dual inverse temperature~$\beta^*$, then
\begin{equation}
\label{eq:kramers-wannier}
\tan \tfrac{1}{2}\theta~=~x~=~x^*~=~\tan \tfrac{1}{2}(\tfrac{\pi}{2}-\theta^*)\,,
\end{equation}
which gives~$\theta_{\mathrm{crit}}=\frac{\pi}{4}$ and~$x_{\mathrm{crit}}=2^{\frac{1}{2}}-1$. Though self-duality a priori does not imply criticality, there are several ways to see that the properties of spin-spin expectations are very different for~$\beta$ above and below $\beta_{\mathrm{crit}}$\,, thus justifying the phase transition. A proof based on the random-cluster representation of the Ising model can be found in~\cite{beffara-duminil-copin-12}. We will also see this in Section~\ref{section:explicit-formulae} when computing the so-called diagonal spin-spin expectations via orthogonal polynomials techniques.
\end{remark}

\subsection{Kac--Ward formula for the partition function}
\label{subsect:kac-ward}
Let~$E(G)$ be the set of \emph{oriented} edges of the graph~$G$ and, for~$e\in E(G)$, let~$\overline{e}$ denote the same edge with the opposite orientation. Further, let us define a matrix~$\mathrm{T}$ indexed by~$E(G)$ as
\[
\mathrm{T}_{e,e'}:=\begin{cases}(x_ex_{e'})^{\frac{1}{2}}\exp[\tfrac{i}{2}\wind(e,e')] & \text{if}~e'~\text{continues}~e; \\
0 & \text{otherwise},
\end{cases}
\]
where in the first line~$e'\ne \overline{e}$ starts at the endpoint of~$e$ and~$\wind(e,e')\in (-\pi,\pi)$ denotes the rotation angle from~$e$ to~$e'$. The famous \emph{Kac--Ward formula}~\cite{kac-ward-52} for the partition function of the Ising model states that
\begin{equation}
\label{kac-ward-formula}
\mathcal{Z}_G = [\det (\mathrm{Id}-\mathrm{T})]^\frac{1}{2}.
\end{equation}

It was an intricate story to give a fully rigorous proof of this identity for general planar graphs (with most of the standard textbooks presenting an incomplete derivation from~\cite{vdovichenko-65}), see~\cite{lis-16} for a streamlined version of classical arguments based on the straightforward expansion of the Kac--Ward determinant.
Another approach (going back to~\cite{hurst-green-60}, see~\cite[Sections~1.3 and~1.4]{chelkak-cimasoni-kassel-16} for historical comments) works as follows. Let~$\mathrm{J}_{e,e'}:=\delta_{\bar{e},e'}$ and~$\mathrm{K}:=\mathrm{J}\cdot(\mathrm{Id}\!-\!\mathrm{T})$, note that the matrix~$\mathrm{K}$ is self-adjoint. {For each~$e\in E(G)$, fix a square root of the direction of $e$ and let~$\eta_e$ be its complex conjugate multiplied by a fixed unimodular factor $\varsigma:=e^{i\frac{\pi}{4}}$. Let~$\mathrm{U}:=\mathrm{diag}(\eta_e)$.}
\begin{theorem}
\label{thm:kac-ward}
The matrix~$\widehat{\mathrm{K}}:={i\mathrm{U}^*\mathrm{K}\mathrm{U}}$ is real anti-symmetric and~$\mathcal{Z}_G\!=\pm \Pf[\,\widehat{\mathrm{K}}\,]$.
\end{theorem}

\begin{remark}
The proof (e.g., see~\cite[Theorem~1.1]{chelkak-cimasoni-kassel-16}) is based on the measure-preserving correspondence between the configurations~$P\in\mathcal{E}_G$ and dimer configurations on some auxiliary non-planar graph~$G^{\mathrm{K}}$ called the~\emph{terminal graph}, whose vertices are in a bijection with~$E(G)$. Note that Theorem~\ref{thm:kac-ward} directly implies~(\ref{kac-ward-formula}). In fact, there exist many other ways to represent the 2D Ising model via dimers (notably a version~\cite{dubedat-bosonization-11} of the classical Fisher mapping onto the dimer model on a planar graph~$G^\mathrm{F}$ constructed from~$G$); see~\cite[Section~3.1]{chelkak-cimasoni-kassel-16} for further discussion.
\end{remark}

\subsection{Fermionic observables}
\label{subsect:fermions}
Given the real anti-symmetric matrix~$\widehat{\mathrm{K}}$, one can introduce Grassmann (i.e., anti-commuting) variables~$(\phi_e)_{e\in E(G)}$ and declare
\[
\Corr{\widehat{\mathrm{K}}}{\phi_{e_1}...\phi_{e_{2k}}}~:=~\mathcal{Z}_G^{-1}\cdot{\textstyle \int} \phi_{e_1}...\phi_{e_{2k}}\exp[-\tfrac{1}{2}\phi^\top\widehat{\mathrm{K}}\phi]d\phi ~=~ \Pf[\,\widehat{\mathrm{K}}^{-1}_{e_p,e_q}\,]_{p,q=1}^{2k}\,.
\]
We need some notation to give a \emph{combinatorial interpretation} of these quantities. Let us add an auxiliary vertex~$z_e$ in the middle of each edge of~$G$ and assign the weight~$x_e^{\frac{1}{2}}$ to both of the half-edges emanating from~$z_e$, which we identify with~$e$ and~$\overline{e}$ according to their orientations. Given a collection~$\mathrm{E}=\{e_1,...,e_{2k}\}\subset E(G)$, let~$\mathcal{E}_G(e_1,...,e_{2k})$ denote the set of all subgraphs~$P$ of this new graph such that the degrees (in~$P$) of all vertices except~$z_{e_1},...,z_{e_{2k}}$ are even, and the following holds for each~$e\in \mathrm{E}$: if~$\overline{e}\not\in \mathrm{E}$, then the degree of~$z_e$ in~$P$ equals~$1$ and~$P$ contains the half-edge identified with~$e$; while if both~$e,\overline{e}\in \mathrm{E}$, then~$z_e$ has \emph{degree~$\mathit{0}$} in~$P$.

\begin{theorem}[see {\cite[Theorem~1.2]{chelkak-cimasoni-kassel-16}}]
\label{thm:multi-fermions} For each set $\{e_1,...,e_{2k}\}\subset E(G)$, one has
\begin{equation}
\label{eq:multi-phi}
\textstyle \Corr{\widehat{\mathrm{K}}}{\phi_{e_1}...\phi_{e_{2k}}} \; =\; \mathcal{Z}_G^{-1}\sum_{P\in\mathcal{E}_G(e_1,...,e_{2k})}x(P)\tau(P)\,,
\end{equation}
where~$x(P)$ denotes the product of all weights of edges and half-edges from~$P$. The sign~$\tau(P)=\pm 1$ is uniquely determined by~$P$ and can be computed as
\begin{equation}
\label{eq:tau-sign-definition}
\textstyle \tau(P):=\sign(s)\cdot\prod_{l=1}^{k} {(i\eta_{e_{s(2l-1)}}\overline{\eta}{}_{e_{s(2l)}})}\exp[-\frac{i}{2}\wind(\gamma_l)]
\end{equation}
if~$P$ is decomposed into a collection of non-intersecting loops and~$k$ paths~$\gamma_l$ running from~$e_{s(2l-1)}$ to~$e_{s(2l)}$, where~$\wind(\gamma_l)$ denotes the total rotation angle of~$\gamma_l$.
\end{theorem}

For an edge~$e$ of~$G$, {introduce a real weight}~$t_e:=(x_e\!+\!x_{e}^{-1})^\frac{1}{2}=(\frac{1}{2}\sin\theta_e)^{-\frac{1}{2}}$.
\begin{definition}
\label{def:observables-edges}
Denote~$\psi(z_e)\!:=\!{t_e\cdot (\eta_e\phi_e+\eta_{\bar{e}}\phi_{\bar{e}})}$. Given two edges~$a$ and~$e$ of~$G$, the \emph{two-point fermionic observables} 
are defined as
\begin{align*}
\Phi_G(a,e):=~\Corr{\widehat{\mathrm{K}}}{t_e\phi_et_a\phi_a}&= \textstyle \mathcal{Z}_G^{-1}\sum_{P\in\mathcal{E}_G(a,e)}t_at_ex(P)({ -i\eta_a\overline{\eta}{}_e}\exp[-\tfrac{i}{2}\wind(\gamma_P)])\,,\\
F_G(a,z_e):=\Corr{\widehat{\mathrm{K}}}{\psi(z_e)t_a\phi_a}
&=\textstyle  \mathcal{Z}_G^{-1}\cdot {(-i\eta_a)}\!\sum_{P\in\mathcal{E}_G(a,z_e)}t_at_ex(P)\exp[-\frac{i}{2}\wind(\gamma_P)]\,,
\end{align*}
where~$\mathcal{E}_G(a,z_e):=\mathcal{E}_G(a,e)\cup\mathcal{E}_G(a,\overline{e})$ and~$\gamma_P$ denotes a path running from~$a$ to~$z_e$ obtained by decomposing~$P$ into a collection of non-intersecting contours.
\end{definition}
\begin{remark} In the \emph{critical} Ising model on the square lattice (or the critical \mbox{Z-invariant} model on an isoradial graph~\cite{boutillier-detiliere-11,chelkak-smirnov-12}), the functions~$F_G(a,z_e)$ are \emph{discrete holomorphic} away from the edge~$a$, see Section~\ref{subsect:s-holomorphicity}. This property was used in~\cite{chelkak-smirnov-12} to prove their convergence to conformal covariant limits as the mesh size tends to zero and in~\cite{hongler-thesis-10,hongler-smirnov-13} to analyze the scaling limit of the energy density field.
\end{remark}

\subsection{Disorder operators}
\label{subsect:spin-disorder} We now describe another approach to 
fermionic observables via the spin-disorder formalism of Kadanoff and Ceva~\cite{kadanoff-ceva-71}. Given \emph{vertices} $v_1,...,v_{2n}$ of~$G$, the correlation of \emph{disorder operators}~$\mu_{v_1},...,\mu_{v_{2n}}$ is defined as
\begin{equation}
\label{eq:disorders-def}
\textstyle \Corr{G}{\mu_{v_1}...\mu_{v_{2n}}}:= \mathcal{Z}_G^{-1}\cdot \mathcal{Z}_G^{[v_1,...,v_{2n}]}\,,\quad \mathcal{Z}_G^{[v_1,...,v_{2n}]}:=\sum_{P\in\mathcal{E}_G(v_1,...,v_{2n})}x(P)\,.
\end{equation}

It is easy to see that~$\mathcal{Z}_G^{[v_1,...,v_{2n}]}$ can be thought of as a partition function of the Ising model defined on the faces of a double-cover~$G^{[v_1,...,v_{2n}]}$ of the graph~$G$ that branches over~$v_1,...,v_{2n}$, with the \emph{spin-flip symmetry constrain}~$\sigma_{u^\sharp}\sigma_{u^\flat}=-1$ for any pair of faces~$u^\sharp$ and $u^\flat$ lying over the same face of~$G$.
One can go further and introduce mixed correlations
\begin{equation}
\label{eq:disorders-spins-mixed-def}
\Corr{G}{\mu_{v_1}...\mu_{v_{2n}}\sigma_{u_1}...\sigma_{u_m}}\;:=\; \Corr{G}{\mu_{v_1}...\mu_{v_{2n}}}\cdot \mathbb{E}_{G^{[v_1,...,v_{2n}]}}[\sigma_{u_1}...\sigma_{u_m}]\,,
\end{equation}
where~$u_1,...,u_m$ should be thought of as faces of the double-cover~$G^{[v_1,...,v_{2n}]}$ described above. By definition of the Ising model on~$G^{[v_1,...,v_{2n}]}$, these quantities obey the sign-flip symmetry between the sheets. It is not hard to see that they admit the following combinatorial interpretation that generalizes both~\eqref{eq:spin-spin-low-temperature} and~\eqref{eq:disorders-def}:
\[
\textstyle \Corr{G}{\mu_{v_1}...\mu_{v_{2n}}\sigma_{u_1}...\sigma_{u_m}} \;=\;
\pm\, \mathcal{Z}_G^{-1}\cdot \sum_{P\in\mathcal{E}_G(v_1,...,v_{2n})}x(P)(-1)^{\mathrm{loops}_{[u_1,...,u_m]}(P\triangle P_0)}\,,
\]
where the~$\pm$ sign depends on the identification of~$u_1,...,u_m$ with faces of~$G$ and~$P_0$ is a fixed collection of edge-disjoint paths matching the vertices~$v_1,...,v_{2n}$ in pairs.
\begin{remark}
Provided~$x_e=x_e^*$, the domain walls representation~\eqref{eq:disorders-def} of disorder correlations coincides with the high-temperature expansion~\eqref{eq:spin-spin-low-temperature} of spin correlations in the dual model. A similar statement holds for mixed correlations: under the Kramers--Wannier duality, disorders are mapped into spins and vice versa.
\end{remark}

Let us now focus on the case when~$m=2n$ and each of the faces~$u_s$ is incident to the corresponding vertex~$v_s$. We call such a pair~$c_s:=(u_s,v_s)$ a~\emph{corner} of the graph~$G$ and attach to~$v_s$ a \emph{decoration} (i.e., a small straight segment oriented from~$u_s$~\emph{towards~$v_s$}) representing this corner. {Let~$\eta_c$ denote the complex conjugate of a square root of the direction of the corresponding decoration, multiplied by $\varsigma$.}

\begin{proposition}[e.g., see~{\cite[Lemma~3.1]{chelkak-cimasoni-kassel-16}}]
\label{prop:multi-chi}
 The following representation holds:
\[
\textstyle \Corr{G}{\mu_{v_1}...\mu_{v_{2n}}\sigma_{u_1}...\sigma_{u_{2n}}}\;=\;\pm\,\mathcal{Z}_G^{-1}\sum_{P\in \mathcal{E}_G(c_1,...,c_{2n})}x(P)\tau(P)\,,
\]
where the set~$\mathcal{E}_G(c_1,...,c_{2n})$ is obtained by attaching decorations~$c_1,...,c_{2n}$ to subgraphs from~$\mathcal{E}_G(v_1,...,v_{2n})$ and the sign~$\tau(P)=\pm 1$ is defined exactly as in~(\ref{eq:tau-sign-definition}).
\end{proposition}

\begin{definition}
\label{def:observables-corners}
Let~$c$ and~$d$ be corners of~$G$. Similarly to Definition~\ref{def:observables-edges}, we set
\[
\textstyle \Phi_G(c,d)\;:=\;\mathcal{Z}_G^{-1}\sum_{P\in\mathcal{E}(c,d)}x(P)({ -i\eta_c\overline{\eta}{}_d}\exp[-\tfrac{i}{2}\wind(\gamma_P)])\,,
\]
where~$\wind(\gamma_P)$ denotes the total rotation angle of a path~$\gamma_P$ running from~$c$ to~$d$ obtained by decomposing~$P$ into a collection of non-intersecting contours.
\end{definition}

\begin{remark}
\label{rem:phi-chi-change} The similarity of combinatorial expansions given in Theorem~\ref{thm:multi-fermions} and Proposition~\ref{prop:multi-chi} allows one to introduce a linear change of Grassmann variables~$\phi_e$ assigned to oriented edges~$e$ emanating from a given vertex~$v$ to a new set of variables~$\chi_c$ labeled by decorations~$c$ attached to~$v$, so that~$\Phi_G(c,d)=\Corr{\widehat{\mathrm{K}}}{\chi_d\chi_c}$ provided that~$c$ and~$d$ are attached to different vertices of~$G$; see~\cite[Section~3.4]{chelkak-cimasoni-kassel-16}. In particular, this implies that~multi-point correlations discussed in Proposition~\ref{prop:multi-chi} satisfy~\emph{Pfaffian identities} similar to multi-point correlations~\eqref{eq:multi-phi}. Further, given a corner~$c$ and an edge~$e$, one can introduce the notation~$\Phi_G(c,e):=\Corr{\widehat{\mathrm{K}}}{t_e\phi_e\chi_c}$ and~$F_G(c,z_e):=\Corr{\widehat{\mathrm{K}}}{\psi(z_e)\chi_c}$. All these observables admit combinatorial representations similar to the ones given in Definition~\ref{def:observables-edges} and Definition~\ref{def:observables-corners}.
\end{remark}

\subsection{S-holomorphicity}
\label{subsect:s-holomorphicity}
From now on, we mostly focus on the case when~$G$ is a subgraph of the square grid and the model is homogeneous, i.e.~$x_e=x=\tan\tfrac{1}{2}\theta$ for all edges of~$G$ and some fixed~$\theta\in(0,\frac{\pi}{2})$. In this case, we always draw a decoration corresponding to a corner~$d=(u,v)$ so that it is directed from the center of the corresponding face~$u=u(d)$ \emph{towards the vertex}~$v=v(d)$.
\begin{definition} We say that a complex-valued function~$F(\cdot)$ defined on mid-edges~$z_e$ of~$G$ and a real-valued function~$\Phi(\cdot)$ defined on corners~$d$ of~$G$ satisfy the \emph{massive s-holomorphicity condition} for a given pair of adjacent~$z_e$ and~$d$, if
\begin{equation}
\label{eq:s-hol-condition}
\Phi(d)\;=\;
\mathrm{Re} [\,e^{\pm \frac{i}{2}(\frac{\pi}{4}-\theta)}{\overline{\eta}_d}F(z_e)\,]\,,
\end{equation}
where the sign is~``$+$'' if~$z_e$ is to the right of~$d$ and~``$-$'' otherwise.
\end{definition}

\begin{remark}
This definition first appeared in~\cite{smirnov-06icm,smirnov-10,chelkak-smirnov-12} {in the critical model context.}
 Note that the papers~\cite{smirnov-10icm,hongler-smirnov-13,chelkak-hongler-izyurov-15} use a slightly different convention for the notion of s-holomorphicity, {which corresponds to the choice $\varsigma=i$ of the global unimodular factor in the definition of~$\eta_e$ and~$\eta_c$.}
\end{remark}

It is well known that observables~$F_G(a,\,\cdot\,),\Phi_G(a,\,\cdot\,)$ or~$F_G(c,\,\cdot\,),\Phi_G(c,\,\cdot\,)$ satisfy~\eqref{eq:s-hol-condition} away from the edge~$a$ or the corner~$c$. This can be deduced both from their combinatorial representations (e.g., see~\cite[Section~4]{smirnov-10icm}) or from the identity
\begin{equation}
\label{eq:mu-mu+s-s=1}
\sin\theta_e\cdot \Corr{G}{\mu_{v^-(e)}\mu_{v^+(e)}\mathcal{O}[\mu,\sigma]}+\cos\theta_e\cdot\Corr{G}{\sigma_{u^-(e)}\sigma_{u^+(e)}\mathcal{O}[\mu,\sigma]}= \Corr{G}{\mathcal{O}[\mu,\sigma]}\,,
\end{equation}
where~$v^\pm (e)$ and~$u^{\pm}(e)$ denote the two vertices and the two faces adjacent to a given edge~$e$ and~$\mathcal{O}[\mu,\sigma]$ stands for an arbitrary product of other disorders and spins, see~\cite[Section~3.6]{chelkak-cimasoni-kassel-16} for more comments on these linear relations.

\begin{remark} For the critical Ising model on the square grid (and similarly for the critical Z-invariant model on an isoradial graph), the factor~$e^{\pm\frac{1}{2}(\frac{\pi}{4}-\theta)}$ in~(\ref{eq:s-hol-condition}) disappears and these equations can be understood as a (strong) form of discrete Cauchy--Riemann equations for the complex-valued function~$F(\cdot)$, see~\cite{smirnov-06icm,smirnov-10,chelkak-smirnov-12}.
\end{remark}

Working with subgraphs of the square grid, let us focus on the real-valued observables~$\Phi_G(c,d)$ restricted onto one of the four possible types of the corners~$d$ and assume that {all the square roots in the definition of~$\eta_d$} are chosen to be the same. It is easy to check (e.g., see~\cite[Lemma~4.2]{beffara-duminil-copin-12}) that, away from~$c$ and the boundary of~$G$, condition~\eqref{eq:s-hol-condition} implies the so-called \emph{massive harmonicity} of~$\Phi_G(c,d)$:
\begin{equation}
\label{eq:massive-harmonicity}
\Delta_\theta\Phi_G(c,d)\;:=\; \Phi_G(c,d)-\tfrac{1}{4}\sin(2\theta){\textstyle \sum\nolimits_{d'\sim d}}\Phi_G(c,d')\;=\;0\,,
\end{equation}
where the sum is taken over four nearby corners~$d'$ of the same type as~$d$.
We will use this equation in Section~\ref{section:explicit-formulae} for explicit computations of the diagonal spin-spin expectations in the full plane via \emph{spinor observables}, to which we now move on.
\begin{remark}
\label{rem:massive=>criticality}
Note that one obtains the same equation for~$\Phi(c,d)$ if~$\theta$ is replaced by~$\theta^*=\frac{\pi}{2}-\theta$ according to the Kramers--Wannier duality~\eqref{eq:kramers-wannier}. It is also worth noting that one can use the massive harmonicity of related fermionic observables arising in the random-cluster representation of the Ising model to prove the criticality of the self-dual value~$\theta_{\mathrm{crit}}=\frac{\pi}{4}$ and to compute the exact rate of the exponential decay of spin-spin expectations in the supercritical model, see~\cite{beffara-duminil-copin-12}.
\end{remark}

\subsection{Double-covers~{$[G;u_1,...,u_m]$} and spinor observables}
\label{subsect:spinors} Given a set of faces~$u_1,...,u_m$ of~$G$, let us fix a collection of paths $\varkappa$ on the graph dual to~$G$ that match these faces (and the outer face if~$m$ is odd) in pairs. Repeating the proof of Theorem~\ref{thm:kac-ward}, it is easy to rewrite~\eqref{eq:spin-spin-low-temperature} as
\[
\textstyle \mathcal{Z}_G\mathbb{E}_G[\sigma_{u_1}...\sigma_{u_m}]\; =\;\pm\,\Pf[{\widehat{\mathrm{K}}_{[u_1,...,u_m]}}]\,,
\]
where~$\widehat{\mathrm{K}}_{[u_1,...,u_m]}=i\mathrm{U}\mathrm{K}_{[u_1,...,u_m]}\mathrm{U}^*$ and the matrix~$\mathrm{K}_{[u_1,...,u_m]}$ is obtained from~$\mathrm{K}$ by replacing its entries~$\mathrm{K}_{e,\bar{e}}=+1$ by~$-1$ if~$e$ crosses one of the paths from~$\varkappa$. Moreover, one has the following analogue of Theorem~\ref{thm:multi-fermions}:
\begin{align*}
\textstyle \Corr{\widehat{\mathrm{K}}_{[u_1,...,u_m]}}{\phi_{e_1}...\phi_{e_{2k}}}
 =\; (\mathcal{Z}_G\mathbb{E}_G[\sigma_{u_1}...\sigma_{u_m}])^{-1}\sum_{P\in\mathcal{E}_G(e_1,...,e_{2k})}x(P)\tau_{[u_1,...,u_m]}(P)\,,
\end{align*}
where~$\tau_{[u_1,...,u_m]}(P):=\tau(P)(-1)^{|P\cap\varkappa|}$; note that this definition depends on~$\varkappa$.

\smallskip

There exists a standard way to make the above construction canonical (i.e. independent of the choice of~$\varkappa$). To do so, let us consider a \emph{double-cover}~$[G;u_1,...,u_m]$ of the graph~$G$ \emph{branching over the faces}~$u_1,...,u_m$ (so that~$\varkappa$ defines its section). Now let us assign Grassmann variables~$\phi_e$ to the edges of~$[G;u_1,...,u_m]$ with the convention~$\phi_{e^\sharp}=-\phi_{e^\flat}$ if~$e^\sharp$ and~$e^\flat$ lie over the same edge of~$G$. One can write
\[
\textstyle \tau_{[u_1,...,u_n]}(P)\,=\,\tau(P)(-1)^{\mathrm{loops}_{[u_1,...,u_m]}(P)} \prod_{l=1}^{k}\mathrm{sheet}_{[G;u_1,...,u_m]}(\gamma_l;e_{s(2l-1)},e_{s(2l)})\,,
\]
where~$\mathrm{sheet}_{[G;u_1,...,u_m]}(\gamma;e,e')$ is equal to $+1$ if the lift of~$\gamma$ onto~$[G;u_1,...,u_m]$ links~$e$ and~$e'$ (considered as the edges on this double-cover) and to $-1$ otherwise. Note that the individual factors may depend on the chosen decomposition of~$P$ into non-intersecting loops and paths but~$\tau_{[u_1,...,u_n]}(P)$ is uniquely determined by~$P$.

\begin{remark} The multi-point fermionic correlators introduced above have a built-in \emph{sign-flip symmetry} between the sheets of~$[G;u_1,...,u_n]$. Functions on double-covers obeying this property are often called \emph{spinors}. As mentioned in Remark~\ref{rem:phi-chi-change} (see also~\cite[Section~3.4]{chelkak-cimasoni-kassel-16}), one can simultaneously use the same notation for the Grassmann variables~$\chi_c$, which are labeled by the corners of~$[G;u_1,...,u_n]$ and can be written as the products~$\mu_{v(c)}\sigma_{u(c)}$ using the language of disorder operators.
In this language, the spinor property of the corresponding observables is a consequence of the similar sign-flip symmetry of mixed correlations~\eqref{eq:disorders-spins-mixed-def}.
\end{remark}

\begin{definition}
\label{def:spinor-observables} Given a set of faces~$u_1,...,u_m$ of~$G$, two distinct corners~$c$ and~$d$ lying on the double-cover~$[G;u_1,...,u_m]$ and a mid-edge~$z_e$ on~$[G;u_1,...,u_m]$, we combinatorially define \emph{spinor observables} with a source at the corner~$c$ by
\begin{align*}
\Phi_{[G;u_1,...,u_m]}(c,d)\;:=\;&\textstyle (\mathcal{Z}_G\mathbb{E}_G[\sigma_{u_1}...\sigma_{u_m}])^{-1}\sum_{P\in\mathcal{E}(c,d)}x(P)\tau_{[u_1,...,u_m]}(P)\,,\\
F_{[G;u_1,...,u_m]}(c,z_e)\;:=\;&\textstyle (\mathcal{Z}_G\mathbb{E}_G[\sigma_{u_1}...\sigma_{u_m}])^{-1}\cdot
{(-i\eta_c)}\sum_{P\in\mathcal{E}(c,z_e)}t_ex(P)\nu(P)\,,
\end{align*}
where~$\nu(P):=\exp[-\tfrac{i}{2}\wind(\gamma_P)]\cdot (-1)^{\mathrm{loops}_{[u_1,...,u_m]}(P)}\cdot \mathrm{sheet}_{[u_1,...,u_m]}(\gamma_P;c,z_e)$. 
\end{definition}

\begin{remark} Similarly to their non-branching counterparts, spinor observables satisfy the massive s-holomorphicity condition~\eqref{eq:s-hol-condition} everywhere on the double-cover~$[G;u_1,...,u_m]$ away from the source corner~$c$ and the boundary of~$G$, including the vicinities of the faces~$u_1,...,u_m$. The proof mimics the case~$m=0$ and can be done, e.g., using the same combinatorial arguments (see~\cite[Section~3.1]{chelkak-hongler-izyurov-15}).
\end{remark}

\subsection{Particular values of spinor observables}
 \label{subsect:spinor-values} Let us now focus on the special situation when the source corner~$c$ is incident to one of the faces~$u_1,...,u_m$. Given a face~$u$, below we denote by~$c=u^{[\eta]}$ one of its corners {such that the corresponding decoration goes in the direction $\varsigma^2\overline{\eta}^2=i\overline{\eta}^2$ and so~$\eta_c=\eta$.} The following lemma shows the relevance of such spinor observables for the analysis of spin correlations.

\begin{lemma}
\label{lemma:Phi-values-near-u}
Let a face~$\widetilde{u}_1$ of~$G$ be such that the corners~$c=u_1^{[\eta]}$ and $d=\widetilde{u}_1^{[i\eta]}$ share a vertex of~$[G;u_1,...,u_m]$. Then the following identity is fulfilled:
\begin{equation}
\label{eq:Phi-cd=ratio-of-spins}
\Phi_{[G;u_1,...,u_m]}(c,d)\;=\;
(\mathbb{E}_G[\sigma_{u_1}\sigma_{u_2}...\sigma_{u_m}])^{-1}\!\cdot\mathbb{E}_G[\sigma_{\widetilde{u}_1}\sigma_{u_2}...\sigma_{u_m}]\,.
\end{equation}
Moreover, if~$\widetilde{u}_1\ne u_2,...,u_m$, then one has $\Delta_\theta \Phi_{[G;u_1,...,u_m]}(c,d)=(\cos\theta)^2$.
\end{lemma}
\begin{proof} A combinatorial proof of the first identity can be found in~\cite[Lemma~2.6]{chelkak-hongler-izyurov-15}. Note that if we write~$\Phi_{[G;u_1,...,u_m]}(c,d)= \Corr{G}{\sigma_{u_1}...\sigma_{u_m}}^{-1}\cdot\Corr{G}{\chi_d\chi_c\sigma_{u_1}...\sigma_{u_m}}^{\vphantom{-1}}$ using the notation discussed above, then~\eqref{eq:Phi-cd=ratio-of-spins} reads simply as $\chi_d\chi_c\sigma_{u_1}=\sigma_{\widetilde{u}_1}$\,. The proof of the second identity is a straightforward computation and the mismatch with~\eqref{eq:massive-harmonicity} is caused by the ambiguity in the definition of~$\Phi_{[G;u_1,...,u_m]}(c,c)=\pm 1$: one should choose different signs in order to fulfill the condition~\eqref{eq:s-hol-condition} for the mid-points~$z_e$ of the two edges incident to~$c$, cf.~\cite[Lemma~3.2]{chelkak-hongler-izyurov-15}.
\end{proof}

The last combinatorial result that we will need is special for the case~$m\!=\!2$. Recall that we denote by~$\mathbb{E}_{G}^*[\sigma_{v_1}\sigma_{v_2}]$ the expectations in the Ising model defined on \emph{vertices} of~$G$, with the inverse temperature~$\beta^*=-\tfrac{1}{2}\log\tan\tfrac{1}{2}\theta^*$ (see~Remark~\ref{rem:kramers-wannier}).

\begin{lemma}
\label{lemma:Phi-m=2-values} Given two corners~$c=u_1^{[\eta]}$ and~$d=u_2^{[\rho]}$, denote by~$v(c)\sim u_1$ and~$v(d)\sim u_2$ the corresponding vertices of~$G$. The following identity is fulfilled:
\begin{equation}
\label{eq:Phi-cd-next-to-u2}
\Phi_{[G;u_1,u_2]}(c,d)\;=\;\pm (\mathbb{E}_G[\sigma_{u_1}\sigma_{u_2}])^{-1}\!\cdot \mathbb{E}_G^*[\sigma_{v(c)}\sigma_{v(d)}]\,,
\end{equation}
Moreover, if~$v(c)\!\not\sim\! u_2$, then~$\Delta_\theta\Phi_{[G;u_1,u_2]}(c,d)=\pm (\sin\theta)^2\cdot \Phi_{[G;u_1,u_2]}(c,d')$, where~$d'$ is the corner of~$u_2$ opposite to~$d$. {The $\pm$ signs depend on the choice of square roots in the definition of~$\eta_c,\eta_d,\eta_{d'}$ and the lifts of the corners $c,d,d'$ onto~$[G;u_1,u_2]$.}
\end{lemma}
\begin{proof} The first claim easily follows from the definition of~$\Phi_{[G;u_1,u_2]}(c,d)$ and the high-temperature expansion~\eqref{eq:spin-spin-high-temperature}, see~\cite[Lemma~2.6]{chelkak-hongler-izyurov-15}. The second is a computation: the mismatch with~\eqref{eq:massive-harmonicity} is now caused by the fact that the two relevant values of~$\Phi_{[G;u_1,u_2]}(c,d')$ correspond to different lifts of~$d'$ onto~$[G;u_1,u_2]$.
\end{proof}

\begin{remark} Contrary to the correlations of the energy density field, one \mbox{cannot} directly represent the spin expectations~$\mathbb{E}_G[\sigma_{u_1}...\sigma_{u_m}]$ neither as the values of fermionic observables nor as the values of their spinor generalizations. Nevertheless, one can use~\eqref{eq:Phi-cd=ratio-of-spins} to control the change of these expectations when moving the faces~$u_1,...,u_m$ step by step. For the \emph{critical} Ising model, this identity was used in~\cite{chelkak-hongler-izyurov-15} as the starting point to deduce the convergence (when the mesh size tends to zero) of spin expectations to conformally covariant limits from the relevant convergence results for \emph{discrete holomorphic} spinor observables, see Section~\ref{section:conformal-invariance} for further details. At the same time, one can use this idea to perform some explicit computations for the infinite-volume limit of the model, as we will see now.
\end{remark}

\section{Diagonal spin-spin expectations in the full plane}
\label{section:explicit-formulae}

\subsection{Setup and preliminaries} The main purpose of this section is to illustrate the general idea of analyzing the spin-spin expectations via spinor observables discussed in Sections~\ref{subsect:spinors} and~\ref{subsect:spinor-values}. Below we work with the infinite-volume limit of the \emph{homogeneous Ising model on the~$\tfrac{\pi}{4}$-rotated square grid}, 
which we denote by~$\Cdiscr$. More precisely, we assume that the centers of faces of~$\Cdiscr$ are located at the points~$(k,s)\in\mathbb{R}^2$ such that~$k,s\in\mathbb{Z}$ and~$k\!+\!s\in 2\mathbb{Z}$.

Recall that (e.g., see~\cite{aizenmann-80}) the 2D Ising model has a unique Gibbs measure above and at the critical temperature, while there are two extremal ones (describing~``$+$'' and~``$-$'' phases, respectively) below criticality. In particular, the infinite-volume limits of the diagonal spin-spin expectations
\[
D_n=D_n(\beta)\;:=\;\mathbb{E}_{\mathbb{C}^\diamond}[\sigma_{(0,0)}\sigma_{(2n,0)}]
\]
are well defined for all~$\beta$ and invariant under translations. Together with the monotonicity of~$D_n$ with respect to~$\beta$, these are the only external inputs that we use below. Our goal is to derive the following classical results from the (massive) harmonicity of spinor observables and their values given by Lemmas~\ref{lemma:Phi-values-near-u} and~\ref{lemma:Phi-m=2-values}.
\begin{theorem}[see~\cite{mccoy-wu-book-14}]
\label{thm:explcit-formulae} For~$\beta=\beta_{\mathrm{crit}}$, the following explicit formula holds true:
\begin{equation}
\label{eq:Dn-crit-explicit}
D_n ~=~ \biggl(\frac{2}{\pi}\biggr)^{\!\!n}\!\cdot\prod_{l=1}^{n-1}\biggl(1-\frac{1}{4l^2}\biggr)^{\!\!l-n}\sim~2^{\frac{1}{3}}e^{-3\zeta'(-1)}\cdot (2n)^{-\frac{1}{4}}\ \ \text{as}\ \ n\to\infty.
\end{equation}
For~$\beta>\beta_{\mathrm{crit}}$, one has~$\lim_{n\to\infty}D_n=(1-q^4)^{1/4}>0$, where~$q=\tan\theta=(\sinh \beta)^{-1}.$
\end{theorem}
\begin{remark}
Since the famous work of Onsager and Kaufman (see~\cite{baxter-11} for historical remarks) it is known that two-point expectations like~$D_n$ can be expressed via Toeplitz determinants, thus the theory of orthogonal polynomials plays a crucial role for their asymptotic analysis. It is worth noting that below we use a shorter route, applying this theory~\emph{directly} to certain polynomials constructed from the values of relevant spinor observables. We believe that one can use this shortcut to study the properties of~$D_n$ in great detail, cf.~\cite[Section~2]{perk-au-yang-09}.
\end{remark}
\begin{remark}
Similar techniques can be applied for the analysis of one-point expectations (with~``$+$'' boundary conditions) in the \emph{``zig-zag'' half-plane}~$\Cdiscr_-$ by which we mean the collection of all faces~$(-k,s)\in\Cdiscr$ with~$k>0$. For instance, for the critical model one can show that~$\mathbb{E}_{\mathbb{C}^\diamond_-}[\sigma_{(-k,\cdot)}]\cdot \mathbb{E}_{\mathbb{C}^\diamond_-}[\sigma_{(-k+1,\cdot)}]=2^{\frac{1}{2}}D_k$\,. This identity leads to an explicit formula for these expectations similar to~\eqref{eq:Dn-crit-explicit}; see~\cite{chelkak-hongler-16}.
\end{remark}

Below we work with a sequence of real-valued spinor observables
\begin{equation}
\label{eq:Theta-n-ks-def}
\Theta_n(k,s)~:=~ D_{n+1} \cdot \Phi_{[\mathbb{C}^\diamond;(-2,0),(2n,0)]}((-2,0)^{[1]},(k,s)^{[i]})\,,\quad n\ge 0\,,
\end{equation}
which should be understood as limits of the similar quantities defined in finite domains~$G$ exhausting~$\Cdiscr$; we assume that these domains are  symmetric with respect to the horizontal axis. It follows from Definition~\ref{def:spinor-observables} and the high-temperature expansion~(\ref{eq:spin-spin-high-temperature}) that~$\mathbb{E}_G[\sigma_{u_1}\sigma_{u_2}]\cdot|\Phi_{[G;u_1,u_2]}(\cdot,\cdot)|\le 1$, so one can use a diagonal process to define all the values~$\Theta_n(k,s)$ as the limits along some subsequence of~$G$.

\subsection{From full-plane spinor observables to orthogonal polynomials} Note that the full-plane observable~$\Theta_n$ has
the following values on the horizontal axis:
\[
\Theta_n(k,0)=0\ \text{if}~k<0\,,\ \ \Theta_n(0,0)=D_n\,;\quad \Theta_n(2n,0)=D_n^*\,,\ \ \Theta_n(k,0)=0~\text{if}~k>2n\,,
\]
where~$D_n^*$ denotes the diagonal spin-spin expectation at the dual temperature. The values~$\Theta_n(0,0)$ and~$\Theta_n(2n,0)$ are essentially given by~\eqref{eq:Phi-cd=ratio-of-spins} and~\eqref{eq:Phi-cd-next-to-u2}. The first and the last claim follow from Definition~\ref{def:spinor-observables} of the spinor observables $\Phi_{[G;(-2,0),(2n,0)]}((-2,0)^{[1]},(k,0)^{[i]})$: if~$P_\pm$ are two configurations contributing to this value that are symmetric to each other with respect to the horizontal axis, then the signs~$\tau_{[G;(-2,0),(2n,0)]}(P_\pm)$ are the same if~$0\le k\le 2n$ and \emph{opposite} otherwise. More generally, if one cuts the double-cover ~$[\Cdiscr;(-2,0),(2n,0)]$ along the horizontal rays~$(-\infty,-2)$ and~$(2n,+\infty)$, then~$\Theta_n$ obeys the symmetry~$\Theta_n(k,-s)=\Theta(k,s)$ on each of the two sheets obtained from~~$[\Cdiscr;(-2,0),(2n,0)]$ in this way.

\begin{remark}
\label{rem:DD*-define-Theta} Applying the maximum principle in the upper half-plane, one can easily see that a \emph{bounded} spinor~$\Theta_n$ symmetric with respect to the horizontal line is uniquely determined by the massive harmonicity property and its values~$D_n$ and $D_n^*$ at the points~$(0,0)$ and~$(2n,0)$ where this property fails.
\end{remark}

Denote 
\[
\widehat{\Theta}_{n,s}(e^{it})~:=~{\textstyle\sum_{k\in\mathbb{Z}:k+s\in 2\mathbb{Z}}}\,e^{\frac{1}{2}ikt}\Theta_n(k,s)\,,\quad s\ge 0\,.
\]
In particular,~$\widehat{\Theta}_{n,0}(e^{it})=D_n+\ldots+D_n^*e^{int}$ is a trigonometric \emph{polynomial}.
The massive harmonicity of~$\Theta_n(k,s)$ in the upper half-plane can be now written as
\begin{equation}
\label{eq:Theta-three-terms}
\widehat{\Theta}_{n,s}(e^{it}) -
(\tfrac{m}{2}\cos\tfrac{t}{2})\cdot (\widehat{\Theta}_{n,s-1}(e^{it})+\widehat{\Theta}_{n,s+1}(e^{it}))=0\,,\quad s\ge 1\,,
\end{equation}
where~$m=\sin(2\theta)=2(q\!+\!q^{-1})^{-1}\!$. Further, Lemmas~\ref{lemma:Phi-values-near-u} and~\ref{lemma:Phi-m=2-values} imply
\[
\Delta_\theta\Theta_n(0,0)=(1\!+\!q^2)^{-1}D_{n+1}\,,\quad \Delta_\theta\Theta_n(2n,0)=(1\!+\!q^2)^{-1}q^2D_{n+1}^*\quad\text{for}\ n\ge 1\,,
\]
and one can similarly check that $\Delta_\theta\Theta_0(0,0)=(1\!+\!q^2)^{-1}(D_1\!+\!q^2D_1^*)$. Together with the symmetry~$\Theta_n(k,-1)=\Theta_n(k,1)$ discussed above, this allows us to write
\begin{equation}\label{eq:Theta-01}
\widehat{\Theta}_{n,0}(e^{it}) - m\cos\tfrac{t}{2}\cdot \widehat{\Theta}_{n,1}(e^{it}) = (1+q^2)^{-1}\cdot(\,\ldots+D_{n+1}+q^2D_{n+1}^*e^{int}+\ldots\,)\,.
\end{equation}
In particular, this trigonometric series does \emph{not} contain monomials~$e^{it},...,e^{i(n-1)t}$; this fact reflects the massive harmonicity of~$\Theta_n$ between the branching points.

\smallskip

Note that one can reverse the above derivation. Namely, given a polynomial~$Q_n(e^{it})=D_n+\ldots+D_n^*e^{int}$, let us define \emph{uniformly bounded} functions
\[
Q_{n,s}(e^{it}):= \biggl[\frac{1-(1-(m\cos\frac{t}{2})^2)^{\frac{1}{2}}}{m\cos\frac{t}{2}}\biggr]^s Q_n(e^{it})\,,\quad s\ge 0\,,
\]
so that~\eqref{eq:Theta-three-terms} holds true for all~$s\ge 1$. Now, if the Fourier series of the function
\begin{equation}
\label{eq:Q-01}
Q_{n,0}(e^{it}) - m\cos\tfrac{t}{2}\cdot Q_{n,1}(e^{it})=(1-(m\cos\tfrac{t}{2})^2)^{\frac{1}{2}}\cdot Q_n(e^{it})
\end{equation}
does not contain monomials~$e^{it},...,e^{i(n-1)t}$, then~$Q_{n,s}$ must coincide with~$\widehat{\Theta}_{n,s}$ due to the uniqueness property of the full-plane observable~$\Theta_n$ described in Remark~\ref{rem:DD*-define-Theta}.

\subsection{Proof of Theorem~\ref{thm:explcit-formulae}} Following the preceding discussion, we are now looking for a trigonometric polynomial~$Q_n(e^{it})=D_n+\ldots+D_n^*e^{int}$, which is orthogonal to all the monomials~$e^{it},...,e^{i(n-1)t}$ with respect to the measure
$\tfrac{1}{2\pi}w(e^{it}){dt}$ on the unit circle, where the \emph{real weight}~$w(e^{it})$ is given by
\[
w(e^{it})=w(e^{-it})\;:=\;(1\!+\!q^2)\cdot(1\!-\!(m\cos\tfrac{t}{2})^2)^{\frac{1}{2}}\,.
\]
For the self-dual value~$\beta=\beta_{\mathrm{crit}}$ we have~$D_n=D_n^*$ and~$q=m=1$, which means~$w(e^{it})=2|\sin\tfrac{t}{2}|$. The above orthogonality condition is now guaranteed if
\[
e^{-\frac{int}{2}}Q_n(e^{it})~=~2D_n\cos\tfrac{nt}{2}+\ldots ~=~ 
P_n(\cos\tfrac{t}{2})\,,
\]
and $\int_{-1}^1P_n(x)x^ldx=0$ for all~$l<n$. In other words,~$P_n(x)=2^nD_nx^n+\ldots$ must be proportional to the~$n$-th Legendre polynomial~$(2^nn!)^{-1}d[(x^2\!-\!1)^n]/dx^n$. Moreover, it follows from~(\ref{eq:Q-01}) and~\eqref{eq:Theta-01} that
\[
\textstyle \int_{-1}^1P_n(x)x^ndx ~=~ \frac{1}{4}\int_0^{2\pi}(2D_{n+1}\cos\frac{nt}{2}+\ldots)(\cos\frac{t}{2})^ndt~=~ \pi 2^{-n}D_{n+1}\,.
\]
Using the well-known expression for the norms of Legendre polynomials, we conclude that
$\pi 2^{-2n}D_{n+1}/D_n= ((2n\!-\!1)!!/n!)^{-2}\cdot 2/(2n\!+\!1)$, which leads to~\eqref{eq:Dn-crit-explicit}.

\smallskip

The subcritical case~$\beta>\beta_{\mathrm{crit}}$ is slightly more involved. Clearly, we should have~$Q_n(e^{it})=c_n\Phi_n(e^{it})+c_n^*\Phi_n^*(e^{it})$, where~$\Phi_n(z)=z^n+\ldots$ is the~$n$-th monic orthogonal polynomial and~$\Phi_n^*(z)=z^n\Phi_n(z^{-1})$; see~\cite[Section~2]{simon-opuc-05} for the notation and basic facts about orthogonal polynomials on the unit circle. For~$n=0$, we simply have~$Q_0(e^{it})=1$ and the Fourier expansion~\eqref{eq:Theta-01} implies
\begin{equation}
\label{eq:D1+^D1*=}
\textstyle D_1+q^2D_1^*\;=\;\beta_0:=\|1\|^2=\frac{1}{2\pi}\int_0^{2\pi}w(e^{it})dt\,.
\end{equation}
For~$n\ge 1$, considering the free term and the highest monomial of~$Q_n(e^{it})$ and using the Fourier expansion~\eqref{eq:Theta-01} of the product~$w(e^{it})Q_n(e^{it})$ we find
\begin{equation}
\label{eq:Dn-recurrence-via-cn}
\begin{array} {lr}
D_n=c_n^*-\alpha_{n-1}c_n\,, &  D_{n+1}=c_n^*\|\Phi_n^*\|^2=c_n^*\beta_n\,, \\
D_n^*=c_n-\alpha_{n-1}c_n^*\,, &\ \ q^2D_{n+1}^*=c_n\|\Phi_n\|^2=c_n\beta_n\,,
\end{array}
\end{equation}
where~$\alpha_{n-1}=\overline{\alpha}_{n-1}:=-\Phi_n(0)$ and $\beta_n:=\|\Phi_n\|^2\!=\|\Phi_n^*\|^2\!=\beta_0\prod_{l=0}^{n-1}(1-\alpha_l^2)$. This allows us to express~$D_{n+1}$ and~$D_{n+1}^*$ via~$D_n$ and~$D_n^*$ for~$n\ge 1$ but unfortunately we cannot extract individual values of~$D_1$ and~$D_1^*$ from~\eqref{eq:D1+^D1*=}. Nevertheless, we can combine~\eqref{eq:Dn-recurrence-via-cn} with the Szeg\"o recurrence relations for the polynomials~$\Phi_n(z)$ and~$\Phi_n^*(z)$ applied at the point~$z=q^2<1$ and obtain the following identity:
\begin{align*}
D_{n+1}\Phi_n^*(q^2)+q^2D_{n+1}^*\Phi_n(q^2) & ~=~ \beta_n\cdot (D_n\Phi_{n-1}^*(q^2)+q^2D_n^*\Phi_{n-1}(q^2)) \\
& ~=~ \dots ~=~ \beta_{n}...\beta_1\cdot (D_1+q^2D_1^*) ~=~ \beta_{n}...\beta_1\beta_0\,.
\end{align*}
The first Szeg\"o theorem~(e.g., see~\cite[Theorems~8.1,8.4]{simon-opuc-05}) implies that~$\beta_n\to (D(0))^2$ and~$\Phi_n^*(q^2)\to D(0)/D(q^2)$ as~$n\to\infty$, where the inner function~$D(z)\!=\!(1-q^2z)^{\frac{1}{2}}$ satisfies~$|D(e^{it})|^2\!=\!w(e^{it})$. Since~$D(0)\!=\!1$, the values~$\Phi_n(q^2)$ are bounded. We know from~\eqref{eq:Dn-crit-explicit} and the monotonicity of~$D_n$  with respect to~$\beta$ that~$D_{n+1}^*\to 0$. Therefore, the second Szeg\"o theorem (e.g., see~\cite[Theorems~8.5]{simon-opuc-05}) gives
\begin{align*}
\lim_{n\to \infty} D_{n+1} = \frac{\lim_{n\to\infty} \beta_n...\beta_1\beta_0}{\lim_{n\to\infty} \Phi_n^*(q^2)} = D(q^2)\exp\biggl[\frac{1}{\pi}\iint_{\mathbb{D}}\,\biggl|\frac{D'(z)}{D(z)}\biggr|^2\!dA(z)\biggr] = (1-q^4)^{\frac{1}{4}}\,.
\end{align*}

\section{Convergence and conformal invariance at criticality}\label{section:conformal-invariance}

In this section we consider the critical Ising model defined on a sequence of discrete approximations to a given bounded planar domain~$\Omega$. For simplicity, below we discuss ``$+$'' boundary conditions only and assume that~$\Omega$ is simply connected; see~\cite{chelkak-hongler-izyurov-16} for the general setup. We denote by~$\Omega_\delta$ a discrete approximation to~$\Omega$ (in Hausdorff or Carath\'eodory sense) on the~\emph{$\frac{\pi}{4}$-rotated square grid~$\Cdiscr_\delta:=\delta\Cdiscr$} with the \emph{mesh size~$2^{\frac{1}{2}}\delta$}. The main object of interest is the asymptotic behaviour of correlation functions such as~$\mathbb{E}_{\Omega_\delta}[\sigma_{u_1}...\sigma_{u_m}]$ in the regime when the points~\mbox{$u_1,...,u_m\in\Omega$} are fixed and~$\delta\to0$, so that the numbers of lattice steps separating these points from each other (and from the boundary of~$\Omega$) are all proportional to~$\delta^{-1}\!\to\infty$. We call this regime a \emph{scaling limit} of the critical Ising model on~$\Omega$; note that one can similarly treat fermionic observables~$\Corr{\widehat{\mathrm{K}}}{\phi_{e_1}...\phi_{e_{2k}}}$ discussed in Section~\ref{subsect:fermions}, spin-disorder correlators~$\Corr{\Omega_\delta}{\mu_{v_1}...\mu_{v_{2n}}\sigma_{u_1}...\sigma_{u_m}}$ from Section~\ref{subsect:spin-disorder}, etc.

\smallskip

In the physics literature (e.g., see~\cite{mussardo-book-10}), the 2D Ising model is considered to be an archetypical example of a discrete system whose large-distance behavior at criticality is prescribed by Conformal Field Theory (with the central charge~$\frac{1}{2}$). In particular, this gives a number of predictions for the scaling limits of correlation functions discussed above, often leading to exact formulae for them. For instance, the CFT counterparts~$\Corr\Omega{\sigma_{u_1}...\sigma_{u_m}}$ of the multi-point spin expectations~$\mathbb{E}_{\Omega_\delta}[\sigma_{u_1}...\sigma_{u_m}]$ have the following explicit form in the upper half-plane~$\mathbb{H}$:
\begin{equation}
\label{eq:spins-explicit}
\Corr{\mathbb{H}}{\sigma_{u_1}...\sigma_{u_m}}= \prod_{p=1}^m(2\mathrm{Im}\,u_p)^{-\frac{1}{8}}\cdot \Biggl[2^{-\frac{m}{2}}\!\!\!\sum_{s\in\{\pm 1\}^m}\prod_{1\le p<q\le m}\biggl|\frac{u_p-u_q}{u_p-\overline{u}_q}\biggr|{\vphantom{\Bigr|}}^{\frac{s_ps_q}{2}}\,\Biggr]^{\frac{1}{2}}
\end{equation}
and are defined in all other simply connected domains~$\Omega$ by the following~\emph{conformal covariance} property under conformal mappings~$\varphi:\Omega\to\Omega'$:
\begin{equation}
\label{eq:covariance-spins}
\textstyle \Corr{\Omega}{\sigma_{u_1}...\sigma_{u_m}}=\Corr{\Omega'}{\sigma_{\varphi(u_1)}...\sigma_{\varphi(u_m)}}\cdot\prod_{1\le p\le m} |\varphi'(u_p)|^{\frac{1}{8}}\,.
\end{equation}
A simpler example is the multi-point correlations of \emph{holomorphic fermions}
\begin{equation}
\label{eq:covariance-fermions}
\textstyle \Corr{\Omega}{\psi_{z_1}...\psi_{z_{2k}}}=\Corr{\Omega'}{\psi_{\varphi(z_1)}...\psi_{\varphi(z_{2k})}}\cdot\prod_{1\le p\le 2k} (\varphi'(z_{2k}))^{\frac{1}{2}},
\end{equation}
which are the CFT counterparts of the fermionic observables~$\Corr{\widehat{\mathrm{K}}}{\psi(z_{e_1})...\psi(z_{e_{2k}})}$ discussed in Section~\ref{subsect:fermions}. In this case, one has~$\Corr{\mathbb{H}}{\psi_{z_1}...\psi_{z_{2k}}} = \Pf[(z_p\!-\!z_q)^{-1}]_{p,q=1}^{2k}$, thus confirming the ``free fermion'' nature of the corresponding field theory.

\smallskip

It is worth noting that Conformal Field Theory assumes existence and conformal covariance of correlation functions in continuum as an axiom, not addressing the proof of convergence of their discrete prototypes as~$\delta\to 0$. In the last several years, such convergence results have been rigorously established for all the primary fields in the Ising model: fermions and energy-densities~\cite{chelkak-smirnov-12,hongler-smirnov-13,hongler-thesis-10}, spins~\cite{chelkak-hongler-izyurov-15}, as well as disorders and mixed correlation of these fields~\cite{chelkak-hongler-izyurov-16}. This progress is based on convergence results for discrete holomorphic observables introduced in Section~\ref{subsect:fermions} (fermions) and Section~\ref{subsect:spinors} (spinors), which are thought of as solutions to discretizations  of special Riemann-type boundary value problems described below.

\subsection{Fermionic observables} \label{subsect:fermions-convergence} Recall that, given an oriented edge~$a$ of~$\Omega_\delta$ and a midpoint of (another) edge of~$\Omega_\delta$, we denote by~$F_{\Omega_\delta}(a,z_e)=\Corr{\widehat{\mathrm{K}}}{\psi(z_e)t_a\phi_a}$ the basic discrete holomorphic observables introduced in Definition~\ref{def:observables-edges}. 
It was shown in~\cite{smirnov-06icm,smirnov-10,chelkak-smirnov-12,hongler-smirnov-13} that such functions can be thought of as unique solutions to discrete boundary value problems whose continuous counterparts we now describe.

\begin{definition}
\label{def:feta}
Given a planar domain~$\Omega$, a point~$a\in\Omega$ and a complex number~$\eta$, we denote by~$f^{[\eta]}_\Omega(a,\cdot)$ the unique function holomorphic in~$\Omega\setminus\{a\}$ such that~$f^{[\eta]}(a,z)={\overline{\eta}}\cdot (z\!-\!a)^{-1}\!+O(1)$ as~$z\to a$ and~$\mathrm{Im}\bigl[f^{[\eta]}_\Omega(a,\zeta)(\tau(\zeta))^{\frac{1}{2}}\bigr]=0$ for all~$\zeta\in\partial\Omega$, where~$\tau(\zeta)$ denotes the (counterclockwise) tangent to~$\partial\Omega$ at~$\zeta$.
\end{definition}

Note that the above boundary conditions can be reformulated as the Dirichlet ones for the harmonic function~$h(z):=\mathrm{Im}\bigl[\int(f_\Omega^{[\eta]}(a,z))^2dz\bigr]$. For~$\eta\!=\!0$, the maximum principle gives~$h\!=\!0$, which justifies the uniqueness for all~$\eta\!\in\!\mathbb{C}$ and also implies that~$f_\Omega^{[\eta]}(a,\cdot)$ depends on~$\eta$ in a \emph{real-linear} manner.
Moreover, one has
\begin{equation}
\label{eq:fermions-in-continuum-def}
f^{[\eta]}_\Omega(a,z) \;=\; {\tfrac{1}{2}(\overline{\eta}f_\Omega(a,z)+ \eta f_\Omega^\star(a,z))\,,}\quad \begin{array}{rcl}f_\Omega(a,z)&\!\!=\!\!&-f_\Omega(z,a)\,,\\
f_\Omega^{{\star}}(a,z)&\!\!=\!\!&-\overline{f_\Omega^{{\star}}(z,a)}\,, \end{array}
\end{equation}
where the function~$f_\Omega(a,z)$ is holomorphic in both variables and has the singularity~$(z\!-\!a)^{-1}$ on the diagonal~$z\!=\!a$, the function~$f_\Omega^{{\star}}(a,z)$ is holomorphic in~$z$, anti-holomorphic in~$a$ and continuous up to~$z\!=\!a$, and~$f_\Omega^{{\star}}(a,\zeta)=\overline{\tau(\zeta)f(a,\zeta)}$ for~$\zeta\in\partial\Omega$; {see~\cite{chelkak-hongler-izyurov-16} for details.} From these properties it is easy to conclude that
\[
\begin{array}{ll}
f_{\Omega}(a,z)=f_{\Omega'}(\varphi(a),\varphi(z))\cdot(\varphi'(a)\varphi'(z))^{\frac{1}{2}},\quad& f_{\mathbb{H}}(a,z)={2}(z-a)^{-1},\cr
f^{{\star}}_{\Omega}(a,z)=f^{{\star}}_{\Omega'}(\varphi(a),\varphi(z))\cdot(\overline{\varphi'(a)}\varphi'(z))^{\frac{1}{2}},& f^{{\star}}_{\mathbb{H}}(a,z)={2}(z-\overline{a})^{-1}
\end{array}
\]
for conformal maps~$\varphi:\Omega\to\Omega'$.

\smallskip

The next theorem was proved in~\cite{hongler-smirnov-13} using techniques from~\cite{chelkak-smirnov-12} (one can drop smoothness assumptions on~$\partial\Omega$ adapting a more robust scheme of the proof from~\cite[Section~3.4]{chelkak-hongler-izyurov-15}). This result also holds true in the isoradial setup, ad verbum.

\begin{theorem}
\label{thm:fermions-convergence} Let two edges~$a_\delta$ and~$e_\delta$ of~$\Omega_\delta$ approximate distinct inner points~$a$ and~$z$ of~$\Omega$, and $\eta=\eta_{a_\delta}$ denote the square root of the direction of~$a_\delta$. One has
\[
\delta^{-1}\cdot F_{\Omega_\delta}(a_\delta,z_{e_\delta}) ~\to~ \tfrac{2}{\pi}f^{[\eta]}_\Omega(a,z)\quad \text{as}\ \ \delta\to 0\,.
\]
\end{theorem}

\begin{remark}
\label{rem:psi-opsi-convergence} Similarly to the {notation}~$\psi(z_a)= {t_a\cdot (\eta_a\phi_a+\eta_{\bar{a}}\phi_{\bar{a}})}$ introduced in Definition~\ref{def:observables-edges}, set~$\opsi(z_a):= {t_a\cdot (\overline{\eta}{}_a\phi_a+\overline{\eta}{}_{\bar{a}}\phi_{\bar{a}})}$. Using~\eqref{eq:fermions-in-continuum-def}, one can rewrite the statement of Theorem~\ref{thm:fermions-convergence} as
\begin{align*}
\delta^{-1}\cdot\Corr{\widehat{\mathrm{K}}}{\psi(z_e)\psi(z_a)}&~=~ \delta^{-1}\cdot({\eta_a}F_{\Omega_\delta}(a,z_e)+{\eta_{\bar{a}}}F_{\Omega_\delta}(\overline{a},z_e)) ~\to~\tfrac{2}{\pi}f_\Omega(a,z)\,,\\
\delta^{-1}\cdot\Corr{\widehat{\mathrm{K}}}{\psi(z_e)\opsi(z_a)}&~=~ \delta^{-1}\cdot ({ \overline{\eta}{}_a}F_{\Omega_\delta}(a,z_e)+{\overline{\eta}{}_{\bar{a}}}F_{\Omega_\delta}(\overline{a},z_e)) ~\to~\tfrac{2}{\pi}f^{{\star}}_\Omega(a,z)\,.
\end{align*}
 This motivates the following \emph{definition}: $\Corr\Omega{\psi_z\psi_a}:=f_\Omega(a,z)$, $\Corr\Omega{\psi_z\opsi_a}:=f_\Omega^{{\star}}(a,z)$, $\Corr\Omega{\opsi_z\opsi_a}:=\overline{f_\Omega(a,z)}$, which can be further extended to multi-point functions such as~$\Corr\Omega{\psi_{z_1}...\psi_{z_{2k}}}:=\Pf[\Corr\Omega{\psi_{z_p}\psi_{z_q}}]_{p,q=1}^{2k}$\,. Theorem~\ref{thm:fermions-convergence} can be then extended to all the multi-point fermionic correlations discussed in Section~\ref{subsect:fermions}. Note that the conformal covariance~\eqref{eq:covariance-fermions} of these scaling limits appears automatically as an intrinsic property of solutions to the boundary value problems from Definition~\ref{def:feta}.
\end{remark}

\subsection{Energy densities} \label{subsect:energies-convergence} For an edge~$e$ of~$\Omega_\delta$, let~$u^\pm(e)$ denote two faces of~$\Omega_\delta$ incident to~$e$. We
introduce a random variable~$\varepsilon_e$ called the \emph{energy density} on~$e$ as
\begin{equation}
\label{eq:energy-density-def}
\varepsilon_e:=(\sin\theta_e)^{-1}\bigl[\sigma_{u^-(e)}\sigma_{u^+(e)}-\tfrac{\pi-2\theta_e}{\pi\cos\theta_e}\bigr] =(\cos\theta_e)^{-1}\bigl[\tfrac{2\theta_e}{\pi\sin\theta_e}-\mu_{v^-(e)}\mu_{v^+(e)}\bigr]\,,
\end{equation}
where the second equality follows from~\eqref{eq:mu-mu+s-s=1} and the multiplicative normalization is chosen so as to remove the lattice-dependent constants from the results given below when working in the isoradial setup; recall that~$\theta_e=\frac{\pi}{4}$ for the square lattice. The additive counterterms correspond to the infinite-volume limit of the model; see~\cite[Corollary~11]{boutillier-detiliere-11} for their exact values. It is well known that one can express all the expectations~$\mathbb{E}_{\Omega_\delta}[\varepsilon_{e_1}...\varepsilon_{e_k}]$ using discrete fermionic observables discussed above (see~Definition~\ref{def:observables-edges} and Remark~\ref{rem:psi-opsi-convergence}). For instance, for~$k=1$ one has
\begin{equation}
\label{eq:energy-via-fermions}
{\tfrac{i}{2}}\Corr{\widehat{\mathrm{K}}}{\psi(z_e)\opsi(z_e)} ~=~ { i\eta_e\overline{\eta}{}_{\bar{e}}}\Phi_{\Omega_\delta}(\overline{e},e)~=~ 
\varepsilon_e^\infty+\mathbb{E}_{\Omega_\delta}[\varepsilon_e]\,,
\end{equation}
where~$\varepsilon_e^\infty=(\sin\theta_e)^{-1}[1+\tfrac{\pi-2\theta_e}{\pi\cos\theta_e}]$. The next result was proved by Hongler and Smirnov~\cite{hongler-smirnov-13} for~$k=1$ and later extended by Hongler~\cite{hongler-thesis-10} to all~$k\ge 1$ (for the square grid case, the generalization to isoradial graphs is straightforward).
\begin{theorem}
\label{thm:energies-convergence} Let a collection of edges~$e_1$,...,$e_k$ of~$\Omega_\delta$ approximate distinct inner points~$z_1,...,z_k$ of a domain~$\Omega$ as~$\delta\to 0$. Then the following is fulfilled:
\[
\delta^{-k}\cdot \mathbb{E}_{\Omega_\delta}[\varepsilon_{e_1}...\varepsilon_{e_k}] ~\to~(\tfrac{2}{\pi})^k\cdot\Corr{\Omega}{\varepsilon_{z_1}...\varepsilon_{z_k}}\quad \text{as}\ \ \delta\to 0\,,
\]
where~$\Corr{\Omega}{\varepsilon_{z_1}...\varepsilon_{z_k}}:=i^k\Corr{\Omega}{\psi_{z_1}\opsi_{z_1}...\psi_{z_k}\opsi_{z_k}}$ and the latter function is defined as the Pfaffian of the corresponding two-point fermionic correlators, see Remark~\ref{rem:psi-opsi-convergence}. In particular, one has the following covariance rule under conformal maps~$\varphi:\Omega\to\Omega'$:
\begin{equation}
\label{eq:covariance-energies} \textstyle \Corr{\Omega}{\varepsilon_{z_1}...\varepsilon_{z_k}}~=~ \Corr{\Omega'}{\varepsilon_{\varphi(z_1)}...\varepsilon_{\varphi(z_k)}}\cdot\prod_{1\le p\le k} |\varphi'(u_p)|\,.
\end{equation}
\end{theorem}
\begin{remark}
According to~\eqref{eq:energy-via-fermions}, in order to prove Theorem~\ref{thm:energies-convergence} one should strengthen Theorem~\ref{thm:fermions-convergence} and analyze the scaling limit of the discrete fermionic observables~$\Corr{\widehat{\mathrm{K}}}{\psi(z_e)\opsi(z_a)}= {t_a\cdot} ({\overline{\eta}{}_a}F_{\Omega_\delta}(a,z_e)+{\overline{\eta}{}_{\bar{a}}}F_{\Omega_\delta}(\overline{a},z_e))$ for~$z_e=z_a$. Contrary to its continuous counterpart, this function is {not} fully discrete holomorphic: after a proper adjustment of its value at~$z_a$, all discrete contour integrals around vertices of~$\Omega_\delta$ vanish, but the ones around two nearby faces~$u^{\pm}(a)$, having opposite signs, do~not. Subtracting an \emph{explicit} counterterm  corresponding to the infinite-volume limit (which scales  as~$\delta^2$ outside of the vicinity of~$a$ and so disappears as~$\delta\to 0$), one obtains a function discrete holomorphic near~$z_a$, for which the convergence at~$z_e=z_a$ can be derived from the convergence in the bulk of~$\Omega_\delta$.
\end{remark}

\subsection{Spinor observables and spatial derivatives of spin correlations} \label{subsect:spinors-convergence} We now move on to the scaling limits of spinor observables~$F_{[\Omega_\delta;u_1,...,u_m]}(u_1^{[\eta]},z_e)$, which are of crucial importance for the analysis of spin correlations due to Lemma~\ref{lemma:Phi-values-near-u}.

\begin{definition} \label{def:g} Given a planar domain~$\Omega$ and a collection~$u_1,...,u_m\in\Omega$ of its distinct inner points, we denote by~$g_{[\Omega;u_1,...,u_m]}(\cdot)$ the unique holomorphic spinor defined on the double-cover~$[\Omega;u_1,...,u_m]$ of~$\Omega$ branching over~$u_1,...,u_m$ that satisfies the
following conditions:~$g_{[\Omega;u_1,...,u_m]}(z)=(z-u_l)^{-\frac{1}{2}}[c_l+O(z-u_l)]$ as~$z\to u_l$, where~$c_1={ e^{-i\frac{\pi}{4}}}$,~$c_2,...,c_m\in{e^{i\frac{\pi}{4}}}\mathbb{R}$, and~$\mathrm{Im}[g_{[\Omega;u_1,...,u_m]}(\zeta)(\tau(\zeta))^{\frac{1}{2}}]=0$ for~$\zeta\in\partial\Omega$.
\end{definition}

\begin{remark}
\label{rem:spinor-covariance} Note that we slightly abuse the notation since~$u_1$ plays a special role in the above definition. The uniqueness of~$g_{[\Omega;u_1,...,u_m]}(z)$ follows from the fact that the similar problem with~\mbox{$c_1=0$} has no nontrivial solution~$g(z)$: the harmonic function~$h(z):=\mathrm{Im}[\int (g(z))^2dz]$ should be bounded near~$u_1$ and bounded from above near~$u_2,...,u_m$, which is in contradiction with the (fixed) sign of its normal derivative on~$\partial\Omega$.
Also, one has
$g_{[\Omega;u_1,...,u_m]}(z)=g_{[\Omega';\varphi(u_1),...,\varphi(u_m)]}(\varphi(z))\cdot(\varphi'(z))^{\frac{1}{2}}$ for conformal maps~$\varphi:\Omega\to\Omega'$; this easily follows from the uniqueness property.
\end{remark}

\begin{theorem}[{\cite[Theorem~2.16]{chelkak-hongler-izyurov-15}}]
\label{thm:spinor-convergence} Let~$u_1,...,u_m$ and~$z$ be distinct inner points of~$\Omega$, below we use the same notation~$u_s$ for a face of~$\Omega_\delta$ approximating the point~$u_s$. Let~$\eta\in{\{1,i,e^{\pm i\frac{\pi}{4}}\}}$ and~$e$ be an edge of~$\Omega_\delta$ approximating the point~$z$. One has
\[
\delta^{-\frac{1}{2}}\cdot F_{[\Omega_\delta;u_1,...,u_m]}(u_1^{[\eta]},z_e) ~\to~(\tfrac{2}{\pi})^{\frac{1}{2}}\cdot g_{[\Omega;u_1,...,u_m]}(z) \quad \text{as}\ \ \delta\to 0\,.
\]
\end{theorem}
Clearly, Theorem~\ref{thm:spinor-convergence} is not enough to analyze the spatial derivatives of spin correlations~$\mathbb{E}_{\Omega_\delta}[\sigma_{u_1}...\sigma_{u_m}]$ via the identity~\eqref{eq:Phi-cd=ratio-of-spins} since one needs to consider the scaling limit of the function~$F_{[\Omega_\delta;u_1,...,u_m]}$ \emph{near the singularity}~$u_1$. This analysis can be performed and the result is provided by the next theorem.

\begin{theorem}[{\cite[Theorem~2.18]{chelkak-hongler-izyurov-15}}]
\label{thm:spin-derivatives-converegnce} With the notation of Theorem~\ref{thm:spinor-convergence}, denote by $\widetilde{u}_1:=u_1{ +2i\overline{\eta}{}^2\delta}$ the next (cornerwise) face to~$u_1$ in the direction of~$u_1^{[\eta]}$. One has
\[
(2\delta)^{-1}\!\cdot\biggl[\frac{\mathbb{E}_{\Omega_\delta}[\sigma_{\widetilde{u}_1}\sigma_{u_2}...\sigma_{u_m}]} {\mathbb{E}_{\Omega_\delta}[\sigma_{u_1}\sigma_{u_2}...\sigma_{u_m}]}-1\biggr]~\to~ \mathrm{Re}[\eta^2\mathcal{A}_\Omega(u_1;u_2,..,,u_m)] \quad\text{as}\ \ \delta\to 0\,,
\]
where~$\mathcal{A}_\Omega(u_1;u_2,...,u_m)$ is defined from the following expansion as~$z\to u_1$:
\[
g_{[\Omega;u_1,...,u_m]}(z)={e^{-i\frac{\pi}{4}}}(z\!-\!u_1)^{-\frac{1}{2}}\cdot [1+2\mathcal{A}_\Omega(u_1;u_2,..,u_m)(z\!-\!u_1)+O(z\!-\!u_1)^2]\,.
\]
\end{theorem}
\begin{remark} It easily follows from the conformal covariance of~$g_{[\Omega;u_1,...,u_m]}(z)$ (see Remark~\ref{rem:spinor-covariance}) that~$\mathcal{A}_\Omega(u_1;u_2,...,u_m)$ is a \emph{pre-Schwarzian form}: one has
\begin{equation}
\label{eq:covariance-A}
\mathcal{A}_\Omega(u_1;u_2,...,u_m)~=~ \mathcal{A}_{\Omega'}(\varphi(u_1);\varphi(u_2),...,\varphi(u_m))\cdot \varphi'(u_1) +\tfrac{1}{8}\cdot(\log\varphi')'(u_1)\,.
\end{equation}
for conformal maps~$\varphi:\Omega\to\Omega'$. Note that the factor~$\frac{1}{8}$ above must coincide with the exponent in~\eqref{eq:covariance-spins}, i.e. with the scaling exponent of the spin field. This gives an explanation for its value that does \emph{not} use explicit computations such as~\eqref{eq:Dn-crit-explicit}.
\end{remark}

\subsection{Spin correlations} \label{subsect:spins-convergence} Let~$u_1,...,u_m$ and~$w_1,...,w_m$ be two collections of points of~$\Omega$. The next result is a simple corollary of Theorem~\ref{thm:spin-derivatives-converegnce}: as~$\delta\to 0$, one has
\begin{equation}
\label{eq:spins-ratios-convergence}
\log\frac{\mathbb{E}_{\Omega_\delta}[\sigma_{w_1}\sigma_{w_2}...\sigma_{w_m}]}{\mathbb{E}_{\Omega_\delta}[\sigma_{u_1}\sigma_{u_2}...\sigma_{u_m}]}
\;\to 
\int\nolimits_{(u_1,...,u_m)}^{(w_1,...,w_m)}\! \mathrm{Re}\biggl[\,\sum_{l=1}^m \mathcal{A}_\Omega(u_l;u_1,...,\widehat{u}_l,...,u_m)du_l\biggr].
\end{equation}
In particular, this differential form must be exact and one can \emph{define} the function~$\Corr\Omega{\sigma_{u_1}...\sigma_{u_m}}$ to be the exponential of its primitive, with an appropriate multiplicative normalization given by~\eqref{eq:spins-decorrelation}. The conformal covariance~\eqref{eq:covariance-spins} of these functions is then a simple corollary of~\eqref{eq:covariance-A} and one can check that the CFT prediction~\eqref{eq:spins-explicit} can be indeed obtained in this way; see~\cite[Appendix~A]{chelkak-hongler-izyurov-15}.

\newpage 

The last ingredient needed to deduce from~\eqref{eq:spins-ratios-convergence} the scaling limits of the expectations~$\mathbb{E}_{\Omega_\delta}[\sigma_{u_1}\sigma_{u_2}...\sigma_{u_m}]$ is provided by discrete counterparts of the asymptotics
\begin{equation}
\label{eq:spins-decorrelation}
\begin{array}{rclr}
\Corr\Omega{\sigma_{u_1}...\sigma_{u_m}} & \sim & \Corr\Omega{\sigma_{u_1}....\sigma_{u_{m-1}}}\cdot\Corr\Omega{\sigma_{u_m}}& \text{as}~\ u_m\to\partial\Omega\,,\\
\Corr\Omega{\sigma_{u_1}\sigma_{u_2}} & \sim & |u_2-u_1|^{-\frac{1}{4}}& \text{as}\ \ u_2\to u_1\in\Omega\,.
\end{array}
\end{equation}
In particular, one can show that~$\lim_{u_2\to u_1}\lim_{\delta\to 0}\mathbb{E}_{\Omega_\delta}[\sigma_{u_1}\sigma_{u_2}]/\, \mathbb{E}_{\mathbb{C}^\diamond_\delta}[\sigma_{u_1}\sigma_{u_2}]=1$ and use~\eqref{eq:Dn-crit-explicit} in order to find the correct normalization of the two-point expectations~$\mathbb{E}_{\Omega_\delta}[\sigma_{u_1}\sigma_{u_2}]$; see~\cite[Sections~2.8~and~2.9]{chelkak-hongler-izyurov-15} for further details.

\begin{theorem}[{\cite[Theorem~1.2]{chelkak-hongler-izyurov-15}}]
\label{thm:spins-convergence} Let~$u_1,...,u_m$ be a collection of inner points of a simply connected domain~$\Omega$. The following convergence holds true:
\[
\delta^{-\frac{m}{8}}\mathbb{E}_{\Omega_\delta}[\sigma_{u_1}...\sigma_{u_m}]~\to~\mathcal{C}_\sigma^m\cdot\Corr\Omega{\sigma_{u_1}...\sigma_{u_m}} \quad\text{as}\ \ \delta\to 0\,,
\]
where~$\mathcal{C}_\sigma=2^{\frac{1}{6}}e^{-\frac{3}{2}\zeta'(1)}$ and the functions~$\Corr\Omega{\sigma_{u_1}...\sigma_{u_m}}$ are given by~\eqref{eq:spins-explicit} and~\eqref{eq:covariance-spins}.
\end{theorem}

\subsection{Mixed correlations in continuum} \label{subsect:mixed-in-continuum}
Our last goal for this note is to discuss a generalization of Theorems~\ref{thm:fermions-convergence},~\ref{thm:energies-convergence} and~\ref{thm:spins-convergence} to mixed correlations of spins, disorders, fermions and energy densities. In this section we list several properties of their expected scaling limits (e.g., see~\cite[Section~14.2.1]{mussardo-book-10}) that allow one to determine them uniquely via solutions to boundary value problems similar to the ones discussed in Definitions~\ref{def:feta} and~\ref{def:g}. We claim (in fact, this claim should be considered as a theorem, see~\cite{chelkak-hongler-izyurov-16}) that there exists a collection of functions~$\Corr\Omega{\mu_{v_1}...\mu_{v_n}\sigma_{u_1}...\sigma_{u_m}}$\,, where~$n$ is even and the points $v_l,u_s\in\Omega$ are pairwise distinct, such that the following \emph{overdetermined} set of conditions is satisfied.

\smallskip

{\bf (I)} Each~$\Corr\Omega{\mu_{v_1}...\mu_{v_n}\sigma_{u_1}...\sigma_{u_m}}$ is a spinor defined on the Riemann surface of the function~$(\prod_{l=1}^n\prod_{s=1}^m(v_l-u_s))^{\frac{1}{2}}$. As some of the points~$v_1,...,v_n$ approach~$u_1,..,u_m$ along the rays~$v_s\!-\!u_s\in {i\overline{\eta}{}_s^2}\mathbb{R}$, where~$|\eta_s|\!=\!1$, there exist real-valued limits
\[
\textstyle \Corr\Omega{\psi_{u_1}^{[\eta_1]}\!...\psi_{u_{k}}^{[\eta_{k}]}\mathcal{O}[\mu,\sigma]}\,:=\,\lim_{v_s\to u_s}\!
|(v_1-u_1)...(v_k-u_k)|^{\frac{1}{4}}\Corr\Omega{\mu_{v_1}\sigma_{u_1}...\mu_{v_k}\sigma_{u_k}\mathcal{O}[\mu,\sigma]},
\]
where~$\mathcal{O}[\mu,\sigma]$ stands for the remaining disorders and spins. Due to the spinor nature of~$\Corr\Omega{\mu_{v_1}...\mu_{v_n}\sigma_{u_1}...\sigma_{u_m}}$, these limits change signs if~$\eta_s$ is replaced by~$-\eta_s$ and are anti-symmetric with respect to the order in which~$\psi$'s are written.

\smallskip

{\bf (II)} The functions~$\Corr\Omega{\psi_{u_1}^{[\eta_1]}\!\!...\psi_{u_{k}}^{[\eta_{k}]}\mathcal{O}[\mu,\sigma]}$ satisfy Pfaffian identities (aka fermionic Wick rules). Moreover, they depend on~$\eta_s$ in a real-linear way, which allows one to introduce the notation ($\mathcal{O}[\psi,\mu,\sigma]$ stands for other fermions, disorders and spins)
\begin{equation}
\label{eq:psieta-as-psi-opsi}
\Corr\Omega{\psi_z^{[\eta]}\mathcal{O}[\psi,\mu,\sigma]}~=~ {\tfrac{1}{2}}\bigl[\,{ \overline{\eta}}\Corr\Omega{\psi_z\mathcal{O}[\psi,\mu,\sigma]}+ {\eta}\Corr\Omega{\opsi_z\mathcal{O}[\psi,\mu,\sigma]}\,\bigr]\,.
\end{equation}
Furthermore, one has the identity~$\overline{\Corr\Omega{\mathcal{O}[\psi,\mu,\sigma]}}=\Corr\Omega{\mathcal{O}[\psi^*\!,\mu,\sigma]}$ by which we mean that each of the symbols~$\psi_z$ on the left-hand side must be replaced by~$\opsi_z$ on the right-hand side and vice versa, with all the other symbols kept unchanged.

\smallskip

{\bf (III)} Each of the functions~$\Corr\Omega{\psi_z\mathcal{O}[\psi,\mu,\sigma]}$ is holomorphic in~$z$ and each of the functions~$\Corr\Omega{\opsi_z\mathcal{O}[\psi,\mu,\sigma]}$ is anti-holomorphic in~$z$. Moreover, one has
\[
\Corr\Omega{\opsi_z\mathcal{O}[\psi,\mu,\sigma]}=\tau(z)\Corr\Omega{\psi_z\mathcal{O}[\psi,\mu,\sigma]}\quad\text{for}\ \ z\in\partial\Omega\,,
\]
where~$\tau(z)$ denotes the counterclockwise tangent vector to the boundary~$\partial\Omega$ at~$z$.

\newpage 

{\bf (IV)} Each of the holomorphic functions~$\Corr\Omega{\psi_z...}$ has the following asymptotics (aka operator product expansions) as~$\psi_z$ approaches the other fields:
\begin{align*}
\Corr\Omega{\psi_z\psi_{z'}...}&=\
{2}(z\!-\!z')^{-1}\!\left[\Corr\Omega{...}+O(|z\!-\!z'|^2)\right]\!,\qquad\Corr\Omega{\psi_z\opsi_{z'}...}=O(1)\,,& z\to z';\\
\Corr\Omega{\psi_z\sigma_u...}&=\ \,
{e^{i\frac{\pi}{4}}}(z\!-\!u)^{-\frac{1}{2}} \left[\Corr\Omega{\mu_u...}\! - 4(z\!-\!u)\partial_u\Corr\Omega{\mu_u...}\!+O(|z\!-\!u|^{2})\right]\!, & z\to u\,;\\
\Corr\Omega{\psi_z\mu_v...}&=
{e^{-i\frac{\pi}{4}}}(z\!-\!v)^{-\frac{1}{2}} \left[\Corr\Omega{\sigma_v...}+4(z\!-\!v)\partial_v\Corr\Omega{\sigma_v...}+O(|z\!-\!v|^{2})\right]\!, & z\to v\,.
\end{align*}
Similar expansions are fulfilled for anti-holomorphic functions $\Corr\Omega{\opsi_z...}$.

\smallskip

{\bf (V)} If we denote~$\Corr\Omega{\varepsilon_u...}:=\lim_{z,z'\to u} {\tfrac{i}{2}}\Corr\Omega{\psi_{z}\opsi_{z'}...}$\,, then one has
\begin{align*}
\Corr\Omega{\sigma_{u'}\sigma_u...}&=|u'\!-\!u|^{-\frac{1}{4}}\left[\Corr\Omega{...}\!+ \tfrac{1}{2}|u'\!-\!u|\Corr\Omega{\varepsilon_{u}...}\!+o(|u'\!-\!u|)\right], & u'\to u;\\
\Corr\Omega{\mu_{v'}\mu_v...}&=|v'\!-\!v|^{-\frac{1}{4}}\left[\Corr\Omega{...} -\tfrac{1}{2}|v'\!-\!v|\Corr\Omega{\varepsilon_{v}...}+o(|v'\!-\!v|)\right], & v'\to v.
\end{align*}

\begin{remark} Provided~$\Corr\Omega{1}=1$, conditions (I)--(V) uniquely determine all the correlators that contain an even number of spins but not the normalization of those containing an odd number of spins.  Similarly to~\eqref{eq:spins-decorrelation}, one can add asymptotics~$\Corr\Omega{\sigma_u...}\sim\Corr\Omega{\sigma_u}\Corr\Omega{...}$ as~$u\to\partial\Omega$ to (I)--(V) in order to fix this issue; see~\cite{chelkak-hongler-izyurov-16} for a further discussion including the \emph{consistency} of these conditions.
\end{remark}

\subsection{Conformal covariance and convergence of mixed correlations} \label{subsect:mixed-convergence}
Following the same lines as in the discussion of conformal covariance of fermionic~\eqref{eq:covariance-fermions} and spin~\eqref{eq:covariance-spins} correlators given above, one can deduce from conditions (I)--(V) that
\[
\textstyle \Corr\Omega{\mathcal{O}_1(z_1)...\mathcal{O}_N(z_N)} = \Corr{\Omega'}{\mathcal{O}_1(\varphi(z_1))...\mathcal{O}_N(\varphi(z_N))}\cdot {\prod_{s=1}^N}\varphi'(z_s)^{\Delta^{\!+}({\mathcal{O}_s})}\overline{\varphi'(z_s)}\,^{\Delta^{\!-}({\mathcal{O}_s})}
\]
for conformal maps~$\varphi:\Omega\to\Omega'$, where each of the symbols~$\mathcal{O}_s$ denotes one of the fields~$\sigma,\mu,\psi,\opsi,\varepsilon$ (so that the total number of~$\mu,\psi$ and $\opsi$ is even) and 
\begin{align*}
(\Delta^{\!+},\Delta^{\!-})(\sigma)=(\Delta^{\!+},\Delta^{\!-})(\mu)=(\tfrac{1}{16}\,,\tfrac{1}{16})\,,&\quad (\Delta^{\!+},\Delta^{\!-})(\varepsilon)=(\tfrac{1}{2}\,,\tfrac{1}{2})\,,\\
(\Delta^{\!+},\Delta^{\!-})(\psi)=(\tfrac{1}{2}\,,0)\,,&\quad (\Delta^{\!+},\Delta^{\!-})(\opsi)=(0\,,\tfrac{1}{2})
\end{align*}
are called the \emph{conformal weights}. Let us also set~$(\Delta^{\!+},\Delta^{\!-})(\psi^{[\eta]}):=(\tfrac{1}{4},\tfrac{1}{4})$; note that according to~\eqref{eq:psieta-as-psi-opsi} one should make a change~$\eta'{}_{\!\!s}:=\eta_s\exp[\frac{i}{2}\arg\varphi'(z_s)]$ when writing a similar covariance rule for correlators involving such fermions.

\smallskip

We now come back to the discrete prototypes of the \emph{real-valued} CFT correlators involving the fields~$\sigma,\mu,\varepsilon$ and~$\psi^{[\eta]}$ with~$\eta\in{\{1,i,e^{\pm i\frac{\pi}{4}}\}}$. In fact, all of them can be written using the spin-disorder formalism introduced in Section~\ref{subsect:spin-disorder}: the energy density~$\varepsilon$ is given by~\eqref{eq:energy-density-def} and the fermion~$\psi^{[\eta]}$ should be thought of as the product~$\chi_c=\mu_{v(c)}\sigma_{u(c)}$, where~$\eta_c=\eta$ (see Remark~\ref{rem:phi-chi-change} and Section~\ref{subsect:s-holomorphicity}; note that the s-holomorphicity condition~\eqref{eq:s-hol-condition} is nothing but the discrete counterpart of~\eqref{eq:psieta-as-psi-opsi}). We conclude this note by the following generalization of Theorems~\ref{thm:fermions-convergence},~\ref{thm:energies-convergence} and~\ref{thm:spins-convergence}.

\begin{theorem}[{see~\cite{chelkak-hongler-izyurov-16}}] Let~$z_1,...,z_N$ be a collection of pairwise distinct points in a planar domain~$\Omega$ and each of~$\mathcal{O}_s$ denote either~$\sigma,\mu,\varepsilon$ or~$\psi^{[\eta]}$ with~$\eta\in{\{1,i,e^{\pm i\frac{\pi}{4}}\}}$. Let~$\Delta:=\sum_{s=1}^N(\Delta^{\!+}(\mathcal{O}_s)\!+\!\Delta^{\!-}(\mathcal{O}_s))$. Then one has
\[
\delta^{-\Delta}\cdot\Corr{\Omega_\delta}{\mathcal{O}_1(z_1)...\mathcal{O}_N(z_N)} ~\to~ \mathcal{C}\cdot \Corr\Omega{\mathcal{O}_1(z_1)...\mathcal{O}_N(z_N)}\quad\text{as}\ \ \delta\to 0\,,
\]
where~$\mathcal{C}\!=\!\prod_{s=1}^N\mathcal{C}_{\mathcal{O}_s}$ and~$\mathcal{C}_{\mathcal{O}_s}$ are given by~$\mathcal{C}_\sigma\!=\mathcal{C}_\mu\!=2^{\frac{1}{6}}e^{-\frac{3}{2}\zeta'(1)}\!$ and~$\mathcal{C}_\varepsilon\!=(\mathcal{C}_{\psi})^2\!=\frac{2}{\pi}$\,.
\end{theorem}

\end{document}